\documentclass[11pt]{article}
\usepackage{appendix, verbatim, cancel, graphicx, wrapfig}
\usepackage{amsthm, amsmath, amssymb}
\usepackage{authblk}
\usepackage{fullpage}
\usepackage{comment}
\usepackage{color}
\usepackage[usenames,svgnames]{xcolor}
\usepackage{caption}
\usepackage{paralist}
\usepackage{hyperref}

\theoremstyle{plain}
\newtheorem{lemma}{Lemma}%[theorem] %removed for lncs%
\newtheorem{theorem}{Theorem}%[theorem] %removed for lncs
\newtheorem{corollary}{Corollary}%[theorem] %removed for lncs
\newtheorem{prop}{Proposition}%[theorem] %removed for lncs
\newtheorem{clm}{Claim}%[theorem]
\newtheorem*{theorem*}{Theorem}
\newtheorem*{prop*}{Proposition}
\newtheorem*{lemma*}{Lemma} %removed for lncs
\renewenvironment{proof}{\noindent {\textbf{Proof }}}{$\Box$ \medskip}
\newenvironment{proofs}{\noindent {\emph{Proof Sketch.}}}{$\Box$ \medskip}

\DeclareMathOperator{\Tr}{Tr}

\newcommand{\EqDef}{:=}
\newcommand{\Eq}[1]{Eq.~(\ref{#1})}

\newcommand{\cc}[1]{\mathcal{#1}}
\newcommand{\HH}[1]{\mbox{\rmfamily\textsc{H--#1}}}
\newcommand{\yes}{\mbox{\rmfamily\textsc{yes}}}
\newcommand{\no}{\mbox{\rmfamily\textsc{no}}}

\newcommand{\id}{\mbox{\rmfamily\textsc{id}}}
\newcommand{\unsat}{\mbox{\rmfamily\textsc{unsat}}}

\newcommand{\poly}{\operatorname{poly}}
\renewcommand{\exp}{\operatorname{exp}}
\renewcommand{\P}{\mathrm{P}}
\newcommand{\NP}{\mathrm{NP}}
\newcommand{\QMA}{\mathrm{QMA}}
\newcommand{\MA}{\mathrm{MA}}
\newcommand{\PP}{\mathrm{PP}}

\newcommand{\xbar}{\overline{x}}
\newcommand{\bx}{\bar{x}}
\newcommand{\sat}{\text{SAT}}

\newcommand{\ket}[2]{| #1 \rangle_{#2}}
\newcommand{\bra}[2]{\langle #1 |_{#2}}
\newcommand{\bk}[2]{\ket{#1}{}\bra{#1}{#2}}
\newcommand{\norm}[1]{\left\| #1 \right\|}
\newcommand\qip[2]{\langle #1 | #2 \rangle}

\newcommand{\SAT}[1]{\mathrm{#1SAT}}
\newcommand{\QSAT}[1]{\mathrm{#1QSAT}}
\newcommand{\HSAT}[1]{\HH{$\SAT{#1}$}}
\newcommand{\HQSAT}[1]{\HH{$\QSAT{#1}$}}
\newcommand{\WIDE}{\mathrm{WIDESAT}}
\definecolor{edcolor}{rgb}{0,0.8,0.3}

\hypersetup{
     colorlinks=true,      		% false: boxed links; true: colored links
     linkcolor=Salmon,        	% color of internal links
     citecolor=blue,            % color of links to bibliography
     filecolor=blue,      		% color of file links
     urlcolor=cyan,           	% color of external links
}

\begin{document}
\setlength{\abovedisplayskip}{0pt}
\setlength{\belowdisplayskip}{0pt}
\setlength{\abovedisplayshortskip}{0pt}
\setlength{\belowdisplayshortskip}{0pt} 
\setlength{\textfloatsep}{1pt}

\title{On the complexity of probabilistic trials for hidden satisfiability problems}
\author[1]{Itai Arad\thanks{email: arad.itai@fastmail.com}}
\author[2]{Adam Bouland\thanks{email: adam@csail.mit.edu}}
\author[2]{Daniel Grier\thanks{email: grierd@mit.edu}}
\author[1,3]{Miklos Santha\thanks{email: miklos.santha@gmail.com}}
\author[1]{Aarthi Sundaram\thanks{email: aarthims@gmail.com}}
\author[4]{Shengyu Zhang\thanks{email: syzhang@cse.cuhk.edu.hk}}
\affil[1]{Center for Quantum Technologies, National University of Singapore, Singapore}
\affil[2]{Massachusetts Institute of Technology, Cambridge, MA USA}
\affil[3]{CNRS, IRIF, Universit\'e Paris Diderot 75205 Paris, France}
\affil[4]{Department of Computer Science and Engineering, The Chinese University of Hong Kong, Shatin, N.T., Hong Kong}
%\date{}

%\keywords{computational complexity, satisfiability problems, trial and error, quantum computing, learning theory}
\maketitle

\begin{abstract}
%\anote{Aarthi}{this needs to be rewritten, but is just a first attempt at an abstract.}
%2SAT can be solved in polynomial time, but 
What is the minimum amount of information and time needed to solve 2SAT? When the instance is known, it can be solved in polynomial time, but is this also possible  without knowing the instance? Bei, Chen and Zhang (STOC '13) considered a model where the input is accessed by proposing possible assignments to a special oracle. This oracle, on encountering some constraint unsatisfied by the proposal, returns only the constraint index. It turns out that, in this model, even 1SAT cannot be solved in polynomial time unless $\P=\NP$. Hence, we consider a model in which the input is accessed by proposing probability distributions over assignments to the variables. The oracle then returns the index of the constraint that is most likely to be violated by this distribution. We show that the information obtained this way is sufficient to solve 1SAT in polynomial time, even when the clauses can be repeated. For 2SAT, as long as there are no repeated clauses, in polynomial time we can even learn an equivalent formula for the hidden instance and hence also solve it. Furthermore, we extend these results to the quantum regime. We show that in this setting 1QSAT can be solved in polynomial time up to constant precision, and 2QSAT can be learnt in polynomial time up to inverse polynomial precision.
%$\epsilon$ in time  $n^{O(\log (1/ \epsilon))}$, 
% and that a 2QSAT instance can be learnt up to precision $\epsilon$ in time polynomial in $\log \epsilon$ and the number of qubits.
\end{abstract}

%\vspace{-0.5em}
\section{Introduction}
%\vspace{-0.5em}
$\SAT{}$ has been a pivotal problem in theoretical computer science ever since the advent of the Cook-Levin Theorem \cite{Cook71, Lev73} proving its $\NP$-completeness. It has a wide array of applications in operations research, artificial intelligence and bioinformatics. Moreover, it continues to be studied under various specialized models such as as random-$\SAT{}$ and building efficient $\SAT{}$ solvers for real-life scenarios. In the complexity theoretic setting, we know that while $\SAT{3}$ is $\NP$-complete \cite{Cook71, Lev73}, $\SAT{2}$ can be solved in linear time \cite{Krom76,EIS76,APT79}. Given the fundamental nature of $\SAT{2}$, in this paper, we consider the following question: 
%\vspace{-0.5em}
\begin{center}\emph{What is the minimum amount of information needed to solve $\SAT{2}$ in polynomial time?}\end{center}
%\vspace{-0.5em}
%The standard way of quantifying the complexity of a problem, such as SAT, is 
%It is well known that when given the list of all clauses in the $\SAT{2}$ instance, one can solve $\SAT{2}$ in polynomial time.
More precisely, what happens if there is no direct access to the problem instance? Are there settings where one can \emph{solve} $\SAT{2}$ without ever \emph{learning} the instance under consideration? We can also pose the same question for the quantum setting where the quantum analogue of $\SAT{}$ ($\QSAT{}$) can be seen as a central problem in condensed matter physics. Complexity theoretically, we know that $\QSAT{2}$ can be solved in linear time \cite{ASSZ15,dBG15} while $\QSAT{3}$ is hard for $\QMA_1$ \cite{GN13}, where $\QMA_1$ is a quantum complexity class analogous to $\NP$.
%Adam: I wouldn't say "the" quantum analogue, since some people would argue QMA is the correct analogue.
%What is the minimum amount of information required to solve $\SAT{2}$ (and in the quantum setting, $\QSAT{2}$) in polynomial time? 
We approach these questions through the ``trial and error'' model. 
In this model, one guesses a solution to an unknown constraint satisfaction problem and tests if it is valid. 
If so, then the problem is solved. Otherwise, the trial fails, and some information about what was wrong with the trial is revealed. This type of problem arises in a number of natural scenarios, in which one has incomplete or limited information about the problem they are trying to solve \cite{BCZ13}. For example, the CSP may be instantiated by a complex biological or physical process to which one does not have access. 

This approach to problem solving was first formalized by Bei, Chen and Zhang \cite{BCZ13}. They considered several types of CSPs and analyzed the computational complexity in the ``unknown input'' setting. Specifically, they consider an oracle model where one can propose solutions to the CSP, and if the solution is not satisfying, then the oracle reveals the identity of a constraint which was violated. For example, if the CSP is an instance of Boolean satisfiability ($\SAT{}$), then after an unsuccessful trial, one may learn that ``clause 7 was violated'', but not anything further. In particular, literals present in clause 7 will not be revealed - only the label of violated clause is known. Furthermore, if there are several violated constraints, then the oracle reveals only one of them, in a possibly adversarial manner. In this paper, we will refer to this as the ``arbitrary violated constraint'' oracle.

This model gives extremely limited access to the instance. In fact, Ivanyos et al. \cite{IKQSS14} showed that even if the underlying CSP is a $\SAT{1}$ instance, accessing it with the BCZ oracle, one cannot determine if it is satisfiable in polynomial time unless $\P=\NP$. This drastically increases the difficulty of deciding a trivial problem like $\SAT{1}$ (assuming $\P\neq \NP$). Interestingly, if there is access to a $\SAT{}$ solver, then $\SAT{1}$ (and even generic $\SAT{}$) in this setting can be solved with polynomially many trials \cite{BCZ13}. So in some sense, their model reveals a sufficient amount of \emph{information} to solve the $\SAT{1}$ instance. However, \emph{decoding} this information requires superpolynomial time (assuming $\P \neq \NP$). In short the information needed to solve the problem is present, but it is not accessible to poly-time algorithms.

In this paper, we ask if there are any meaningful modifications of their model which allow us to solve simple CSPs like $\SAT{1}$ and $\SAT{2}$ in polynomial time. A natural starting point is to randomize the ``arbitrary violated constraints'' model. One obvious way to do that is to consider allowing randomized queries to the %arbitrary violated constraint 
oracle. This however does not significantly decrease the complexity of the problems. 
A second approach to randomize is to let the oracle return a violated clause at random.
Contrary to the previous approach, this model trivializes the problem, since
by repeating the same trial many times the oracle will reveal all violated clause indices with high probability. 
This in turn allows one to learn the entire instance, and therefore trivially, to solve 
$\SAT{1}$ and $\SAT{2}$.

Motivated by these unfruitful approaches 
we consider a model which does not allow one to completely learn the underlying instance,
but it still yields polynomial time algorithms for $\SAT{1}$ and $\SAT{2}$.
%his model is interesting in that it still yields poly-time algorithms for $\SAT{1}$ and $\SAT{2}$. 
Specifically, in this model one can propose a probability distribution $D$ over assignments, and the oracle reveals the index of the clause which is most likely to be violated by this trial. If there are multiple clauses with the same probability of violation under $D$, then the oracle can break ties arbitrarily. In particular, product distributions over the variables suffice for our application, so one merely specifies the probability $p_i$ that each variable $x_i$ is set to 1 in the assignment, to $1/\poly$ precision. We show that in this model, there exist cases where one cannot learn the underlying $\SAT{1}$ or $\SAT{2}$ instance. However, despite this limitation, one can still solve 
in polynomial time  $\SAT{1}$ and a restricted version of $\SAT{2}$ where clauses are not repeated. 
In the course of the algorithm for the restricted version of $\SAT{2}$, we actually learn an equivalent formula with the same set of satisfying assignments. Furthermore, we are able extend this model to the quantum setting, 
and show that one can {solve}, 
in polynomial time, Quantum 1SAT ($\QSAT{1}$) up to constant precision.
We also show that in polynomial time we can
{learn} Quantum 2SAT ($\QSAT{2}$) up to inverse polynomial precision. This, however, seems insufficient
to solve the hidden instance in polynomial time due to some subtle precision issues, which we discuss in Section \ref{sec:2qsat}.
%it is possible that 
%in a natural quantum analogue of the model (except in some pathological cases). 

%The seemingly arbitrary modifications discussed above can be motivated as different attempts to randomize  One ob
%\subsection{Relation to prior work}
\vspace{-0.5em}
\subparagraph*{Relation to prior work.} As previously mentioned, Bei Chen and Zhang \cite{BCZ13} introduced the trial and error model. They considered several examples of CSPs and analyzed their complexity under the unknown input model with the ``arbitrary violated constraint'' oracle. With regards to $\SAT{}$, they showed an algorithm to solve hidden-$\SAT{}$ using polynomially many queries to the oracle (given access to a SAT oracle). Furthermore, they showed that one cannot efficiently learn generic SAT instances in this model, because it takes $\Omega(2^n)$ queries to the oracle to learn a clause involving all $n$ variables of the instance. 
%, where given a trial, an oracle reveals an index of a violated clause

Subsequently, Ivanyos \emph{et al.} \cite{IKQSS14} characterized the complexity of classical CSPs in several hidden input models. In particular, they consider the ``arbitrary violated constraint'' model described above, as well as models which reveal more information such as the variables involved in the violated clause or the relation of the violated clause. They show a generic ``transfer theorem'' which classifies the complexity of hidden versions of CSPs given properties of the base CSP. In particular, their transfer theorem implies that the hidden version of $\SAT{1}$ with arbitrary violated constraints cannot be solved in polynomial time unless $\P=\NP$. This indicates that the ``arbitrary violated constraint'' model is fairly restrictive.

In parallel, Bei, Chen and Zhang \cite{BCZ13b} considered a version of the trial and error model for linear programming. Suppose you have a linear program, and you are trying to determine whether or not it is feasible (By standard reductions this is as difficult as solving a generic LP). They consider a model in which one can propose a point, and the oracle will return the index of an arbitrary violated constraint (half-plane) in the linear program. They show that in this model, one requires exponentially many queries to the oracle to determine if an LP is feasible. However, they then consider a relaxation of this model, in which the oracle returns the index of the \emph{worst-violated} constraint, i.e. the half-plane which is furthest (in Euclidean distance) from the proposed point. Surprisingly, they show (using a variant of the ellipsoid algorithm) that one can still solve linear programs in this model in polynomial time. Our model can be seen as an analogue of the ``worst violated constraint'' model of Bei, Chen and Zhang \cite{BCZ13b} for the case of hidden $\SAT{}$ ($\HSAT{}$). 

\vspace{-0.5em}
\subparagraph*{Our Results.} Our results can be broken into several sections. First, we consider a relaxation of the ``arbitrary violated constraint'' model of Bei, Chen and Zhang \cite{BCZ13}, in which the oracle reveals which subset of clauses are violated by each assignment\footnote{This is equivalent to a model in which the oracle reveals a random violated clause - by repeating each query many times one can learn the set of violated clauses with high probability.}. We show that in this model, there exist simple algorithms to learn $\HSAT{1}$ or $\HSAT{2}$ instances, and hence solve them in polynomial time. In some sense these models are almost ``too easy'' as they allow you to easily learn the instance (See Section~\ref{sec:easy}).

We then explore the ``worst violated constraint'' model for the rest of the paper. We provide a toy example as to why this model is more powerful than the ``arbitrary violated constraint'' model of Bei, Chen and Zhang \cite{BCZ13}. They showed that it requires $\Omega(2^n)$ time to learn a $\SAT{}$ clause involving all $n$ variables. In Section~\ref{sec:wide} we show that it is possible to learn such a $\SAT{}$ clause in polynomial time in the ``worst violated constraint'' model. Among our main results is the analysis of the computational complexity of $\HSAT{1}$ and that of $\HSAT{2}$. 
%\vspace{-0.5em}
\begin{theorem*}[Informal statement]
Given a hidden $\SAT{}$ formula $\Phi$ on $n$ variables and $n$ clauses, it is possible to find a satisfying assignment for $\Phi$ in polynomial time if $\Phi$ is a
\begin{inparaenum}[\upshape(\itshape a\upshape)]
\item $\SAT{1}$ formula or
\item $\SAT{2}$ formula with no repeated clauses.
\end{inparaenum}
\end{theorem*}
%\vspace{-0.5em}
%In Section~\ref{sec:1SAT} we show that $\HSAT{1}$ can be solved in polynomial time. 

Our algorithm for $\HSAT{1}$, in Section~\ref{sec:1SAT}, works even when clauses are repeated multiple times in the instance, despite the fact that it's not possible to \emph{learn} the instance in this setting. This is in sharp contrast to the ``arbitrary violated constraint'' model, where even  $\HSAT{1}$ cannot be solved in polynomial time unless $\P=\NP$ \cite{IKQSS14}. The main difficulty in deriving our algorithm for $\HSAT{1}$ comes from dealing with repeated clauses, which allow the oracle to obscure information about the instance. Unlike the $\HSAT{1}$ case, the algorithm for $\HSAT{2}$ discussed in Section~\ref{sec:2SAT}, works by attempting to learn the instance; it either succeeds in learning an equivalent instance (in which case one can solve the problem using any $\SAT{2}$ algorithm), or it accidentally stumbles upon a satisfying assignment in the meantime and aborts. The problem of solving $\HSAT{2}$ with repeated clauses similar to $\HSAT{1}$ is left for future work.

Following this we generalize these results to the quantum case. In this case the goal is to determine if a set of $1$-qubit or a set of $2$-qubit projectors is mutually satisfiable or not. We consider an analogue of this model in which one can propose a probability distribution over quantum states (i.e. a density matrix), and the oracle returns the index of the clause which is most likely to be violated. Our results for hidden $\QSAT{}$ ($\HQSAT{}$) show that
%\vspace{-0.5em}
\begin{theorem*}[Informal statement]
Given a $\HQSAT{}$ instance $H$ defined on $n$ qubits with $m$ projectors and $\epsilon > 0$, 
it is possible to
%\begin{inparaenum}
\begin{enumerate}[(a)]
\item solve $H$ to a precision $\epsilon$ in time $O(n^{\log(1/\epsilon)})$ if $H$ is a $\QSAT{1}$ instance and
\item learn each projector of $H$ up to precision $\epsilon$ in time $O(n^4 + n^2 \log(1/\epsilon))$, if $H$ is a $\QSAT{2}$ instance as long as the interaction graph of $H$ is not star-like.
\end{enumerate}
%\end{inparaenum}
\end{theorem*}
%\vspace{-0.5em}

By star-like, we mean the interaction graph contains an edge that is incident to all other edges in the graph. At this point it is worth comparing the notions of \emph{learning} and \emph{solving} hidden instances both in the classical and quantum settings. The classical case is more straightforward where learning an instance means learning all the literals present in each clause, whereas solving means finding a satisfying assignment. For example, our algorithm for $\HSAT{2}$ without repetitions learns the instance, while our algorithm for $\HSAT{1}$ solves the instance without learning it. For hidden versions of $\SAT{1}$ and $\SAT{2}$, learning the instance in polynomial time automatically triggers solving it in polynomial time as well. 

However, in the quantum setting this simple relation between learning and solving breaks down. 
The continuous nature of $\QSAT{}$ means we can only learn a projector or find a satisfying assignment up to a specified precision $\epsilon$. The latter is accomplished with our $\HQSAT{1}$ algorithm in Section~\ref{sec:1qsat}.
However in the case of hidden $\QSAT{2}$ learning the instance up to precision $\epsilon$ does not imply that one can solve the instance up to precision $\operatorname*{poly}(n,\epsilon)$ in polynomial time. This can be attributed to current algorithms for $\QSAT{2}$ being very sensitive to precision errors. This issue of divergence between the notions of learning and solving $\HQSAT{2}$ instances is further discussed in Section~\ref{sec:2qsat}.

%\vspace{-1.25em}
\section{Notations and Preliminaries}
\vspace{-0.5em}
%\subsection{Boolean SAT}
%\vspace{-2mm}
\subparagraph*{Boolean Satisfiability.} The Boolean satisfiability problem, generally referred to as $\SAT{}$, is a constraint satisfaction problem defined on $n$ variables $\mathbf{x} = \{x_1, \ldots, x_n\}$ where we are given a formula represented as a conjunction of $m$ clauses and each clause is a disjunction of \emph{literals} (variables, $x_j$, or negated variables, $\xbar_j$). The problem is solved if we can find an assignment to the variables (i.e. $\forall \; i, \; x_i \in \{0, 1\}$) that sets the value of every clause to $1$. In particular, if each clause involves at most $k$ literals, then this problem is classified as $\SAT{k}$. It is well known that while $\SAT{2}$ can be solved in linear time~\cite{Krom76,EIS76,APT79}, $\SAT{k}$ for $k \geq 3$ is $\NP$-complete~\cite{Cook71,Lev73}. A useful notion is that of \emph{clause types} which is defined as the unordered set of literals present in the clause. Specifically, the clause type for $C_j = (x_a \vee \xbar_b \vee x_c)$ is denoted by $T(C_j) = \{x_a, \xbar_b, x_c\}$. So, all possible clause types for 
$\SAT{2}$ would be $\left\{ \{x_a, x_b\}, \{x_a, \xbar_b\}, \{\xbar_a, x_b\}, \{\xbar_a, \xbar_b\} \; | \; a, b \in [n] \text{ and } a \neq b \right\}$,
%\end{align*}
where $[n]$ denotes the set $\{1, \ldots, n\}$. From this
definition, it is clear that $\SAT{2}$ has $O(n^2)$ clause types and
similarly, $\SAT{k}$ would have $ \binom{2n}{k} = O(n^k)$ clause
types. . Given a $\SAT{}$ formula $\phi$, we say that the $\SAT{}$
formula $\phi'$ is \emph{equivalent} to $\phi$ if for all
assignments $\mathbf{x} \in \{0, 1\}^n$, $\mathbf{x}$ satisfies
$\phi$ if and only if it satisfies $\phi'$. For any formula $\phi$, $\SAT{}(\phi) \EqDef \{\mathbf{x} \in \{0, 1\}^n \;\; | \;\; \phi(\mathbf{x}) = 1 \}.$
%\begin{align*}
%
%\end{align*}
%\footnote{In line with standard notations, by setting $x_i = 0$, implicitly we also set $\xbar_i = 1$.}

%\subsection{Hidden SAT}
\vspace{-0.5em} 
\noindent 
\subparagraph*{Hidden SAT.} While considering the
unknown input version of $\SAT{}$ (resp. $\SAT{k}$), the boolean
formula is considered as hidden and accessible only via an oracle
that accepts an assignment and reveals some form of violation
information. In our case, this is the ``worst violated oracle'' which
accepts a \emph{probabilistic} assignment and reveals a clause that
has the \emph{highest probability} of being violated with ties being
broken arbitrarily. A probabilistic assignment for a set of $n$
variables is a function $\mathbf{a}: [n] \rightarrow [0, 1]$ such
that $Pr[x_i = 1] = \mathbf{a}(i)$ and $Pr[x_i = 0] = Pr[\xbar_i =
1] = 1 - \mathbf{a}(i)$. For the sake of concise notation, these are
usually written as $x_i = \mathbf{a}(i)$ and $\xbar_i = 1 -
\mathbf{a}(i)$. This naturally translates to the notion of the
probability of a clause $C_j$ being violated which is defined as 
$Pr[C_j = 0] \EqDef \prod_{\ell \in T(C_j)} Pr[\ell = 0] = \prod_{\ell \in T(C_j)} (1 - \ell)$
%\vspace{-1em} 
%\begin{align*}
%
%\end{align*}
which allows the oracle to calculate the probability for each clause
being violated. Here we are using $\ell$ to refer both to the
identity of a literal as well as to the probability that literal
$\ell$ is set to true. Now, the problem $\HSAT{}$ (resp. $\HSAT{k}$)
consists of finding a satisfying assignment for a hidden $\SAT{}$
(resp. $\SAT{k}$) formula by proposing probabilistic assignments to
the ``worst violated oracle''. One way we do this is also by
\emph{learning} an \emph{equivalent formula} to the hidden instance
and solve it to find a satisfying assignment. By learning we mean
the process of using the information from a series of violations to
determine what a clause in the hidden instance could be.

Note that it's possible for an instance to contain clauses which will never be returned by the oracle. For instance, given clauses $C_i$ and $C_j$, if $T(C_i) \subset T(C_j)$, then clause $C_i$ will always be at least as violated as $C_j$. Hence the oracle might never return clause $C_j$. For this reason we will say that $C_i$ \emph{obscures} $C_j$ if $T(C_i) \subset T(C_j)$. An obscured clause might never be returned by the oracle.

% For example, suppose $C_1 = (x_1), C_2 = (x_1 \vee x_2)$. Then  we have $(1 - x_1) > (1 - x_1) (1 - x_2)$ implying that whenever $C_2$ is violated, $C_1$ is also violated and will always be the favoured violation returned by the oracle. 

The complexity of the algorithms in the following sections is in terms of the total running time where one query to the oracle takes unit time.
%\anote{Aarthi}{Will add definitions for all oracle types considered in this paper here.}

%We now formally describe the ``worst violated constraint" model.
%\vspace{-1.25em}

\section{An easier model: an oracle which reveals all violated constraints}
\label{sec:easy}
We begin by considering a simple generalization of the oracle model of Bei Chen and Zhang \cite{BCZ13}. 
We will call this the ``all violated constraints'' model.
In particular, suppose that you have a $\SAT{}$ instance which is hidden from you.
Instead, you have have access to an oracle which, given an assignment to the variables, identifies which clauses are violated by that assignment.
In the model of Bei Chen and Zhang \cite{BCZ13}, the oracle only returns one violated clause, and the oracle's response may be adversarial.
Here, in contrast, the oracle reveals the identities of \emph{all} violated clauses.
This is a more natural model than the ``arbitrary violated constraint'' model, because in many real-life applications, you would expect to see a random violated constraint\footnote{Note that a model which returns a random violated constraint is just as powerful to this one, since by repeating the experiment one could quickly learn the set of all violated constraints with high probability.}, or all constraints, not an adversarially chosen one.
%For instance, if you are testing antibiotics against bacteria, after trying a combination of compounds, you would expect to see the set of all bacteria remaining, not an arbitrarily chosen one.
Therefore in many cases this is a more natural model than the one considered by Bei Chen and Zhang.

We begin by showing that this model is much more powerful than the
``arbitrary violated constraint'' model of BCZ.  In particular, this
model allows one to learn a hidden $\SAT{k}$ instance in $O(n^k)$
time. Therefore, by simply learning the underlying instance, one can
solve $\HSAT{1}$ or $\HSAT{2}$ in polynomial time.

\begin{theorem} In the ``all violated constraints'' model, there is an algorithm which either learns an arbitrary $\SAT{k}$ instance on $n$ variables and $m$ clauses, or else finds a satisfying assignment to the instance, in time $O(m n^k)$, where the big-O notation hides a constant which depends on $k$.
\end{theorem}
\begin{proof}
Let $x_1, x_2, \ldots x_n$ be the variables of your instance, and $C_1 \ldots C_m$ be the clauses. In the following algorithm we will assume that all the assignments tried, fail to satisfy the instance; if they happen to satisfy it, then we have found a satisfying assignment and the algorithm aborts.

The type of a clause is the subset of literals contained in the clause. For instance, a clause could be of type $(x_1\vee x_2 \vee \xbar_3)$ or of type $(\xbar_1 \vee x_3)$. In an instance of $\SAT{k}$, there could be multiple clauses of the same type. We will show that for each clause type $T$ involving $k$ literals, we can learn which clauses of the instance are of type $T$ in time $O(m)$. This implies the claim.
%Note that a $\SAT{k}$ instance can consist of at most $O(n^k)$ clause types. 

We will now show how to learn which clauses are of type $(x_1 \vee x_2 \vee ... \vee x_k)$ in our instance. An analogous proof holds for other clause types. First propose the assignment $x_1=x_2=\ldots x_k=0$ and $x_{k+1}=x_{k+2}=\ldots=x_n=0$. This returns some subset of violated clauses $S$. Next propose the assignment $x_1=x_2=\ldots x_k=0$ and $x_{k+1}=x_{k+2}=\ldots=x_n=1$. This returns a subset $S'$ of violated clauses. Now take the intersection of $S\cap S'$ (which can be done in $O(m)$ time). 

We now claim that $S\cap S'$ contains all clauses $C_j$ for which all the literals in $C_j$ are the set $\{x_1,x_2, \ldots x_k\}$. To see this, first note that any clauses that are in $S\cap S'$ are clearly violated by both of the proposed assignments. Now consider any clause $C_j$ which involves variables outside of $x_1\ldots x_k$. For example say clause $C_j$ contains the literal $x_{k+1}$. Then $C_j$ will be satisfied by one of the two proposed assignments; hence it will not be in $S\cap S'$.
%%So clearly such clauses are in $S\cap S'$.

We therefore have that $S\cap S'$ contains all clauses on the literals $\{x_1 \ldots x_k\}$. 
Some of these clauses are of type $(x_1 \vee x_2 \vee ... \vee x_k)$, but others may involve subsets of these literals.
Now to learn which of these are of type $(x_1 \vee x_2 \vee ... \vee x_k)$, for each subset $L$ of the literals $\{x_1 \ldots x_k\}$ which is of size $k-1$, perform the same experiment to learn which clauses $S_L$ involve literals in $L$ only. Then we have that $\cup_L S_L$ are all clauses on the literals $\{x_1 \ldots x_k\}$ which involve $k-1$ of those literals or fewer. (There are $k$ such sets $L$, each of which takes $O(m)$ time). Hence $\left(S\cap S'\right) \setminus \left(\cup_L S_L\right)$ is the set of all clauses of type $(x_1 \vee x_2 \vee ... \vee x_k)$ as desired.

Therefore for each clause type, it takes $O(m)$ time to learn which clauses are of that type (where we have suppressed a constant depending on $k$). As there are $O(n^k)$ possible clause types in a $\SAT{k}$ instance, this implies one can learn $\SAT{k}$ for any fixed $k$ in time $O(mn^k)$ as desired.
\end{proof}
\vspace*{-0.5em}
\begin{corollary} In the ``all violated constraints'' model, $\SAT{2}$ can be solved in polynomial time. 
\end{corollary}

In some sense this model is too easy. We therefore turn our attention to a more restrictive model: the ``worst violated constraint'' model. Here we allow one to query probability distributions over assignments, and the oracle will return a clause which is the most likely to be violated by that distribution. The oracle may break ties arbitrarily. This can be seen as a CSP-version of the linear programming model of Bei Chen and Zhang \cite{BCZ13b}. %We show that despite giving much less information, this model still admits poly-time solutions of $\SAT{1}$ and $\SAT{2}$ except in certain pathological cases.

\section{Comparison to the Bei, Chen and Zhang Model: WIDESAT}
\label{sec:wide}
Recall that in \cite{BCZ13}, a solution to a hidden formula can be found in $O(nm)$ time where $n$ is the number of variables and $m$ is the number of clauses.  However, this does not mean that an equivalent instance of the hidden formula can be found in polynomial time.  In fact, they show that any randomized algorithm necessarily requires exponentially many queries in order to generate a formula equivalent to the hidden one.  Let a \emph{$\WIDE$ clause} be any clause that contains all $n$ distinct variables.  It turns out that $\WIDE$ clauses are exactly the types of clauses that are difficult to learn in the model of Bei, Chen, and Zhang.\footnote{In fact, this is not hard to see. Suppose the hidden formula has exactly one $\WIDE$ clause.  If the oracle always returns 'YES' to a proposed assignment, then each query only eliminates 1 out of exponentially many $\WIDE$ clause types.} In contrast, we show that similar $\WIDE$ instances can be learned in polynomial time in our model.

\begin{prop}
\label{prop:wide}
Given a hidden $\WIDE$ instance on $n$ variables and $m$ distinct clauses where $m \le n$, we can learn an equivalent instance in $O(\binom{n}{m-1}  2^{m} + n)$ time.
\end{prop}

Before we prove the proposition in full, let us gain some intuition as to why this task should be easier in our model.  Suppose first that there was exactly 1 $\WIDE$ clause containing variables $x_1, x_2, \ldots, x_n$.  Propose the solution $(1, .5, .5, \ldots, .5)$.  If the oracle returns that the formula is satisfied, then $x_1$ must be the literal present in the clause, otherwise if it's not satisfied, then $\overline{x}_1$ must be the literal present in the clause.  Continuing for each separate variable completes the argument.

We will need the following lemma for the full generalization.
\begin{lemma}
\label{lem:wide}
Suppose we have a set of $m$ distinct binary strings of length $n$.  Then there exists a subset of $m-1$ indices such that each string restricted to those indices is unique.\footnote{To see that $m-1$ indices are necessary, consider any subset of size $m$ of $\{e_i \mid 1 \le i \le n\}$ where $e_i$ denotes the binary string with a 1 in position $i$ and 0 everywhere else.  It is clear that any $m-2$ indices are insufficient to distinguish all strings.}
\end{lemma}
\begin{proof}
We proceed by induction. Consider any two strings in the set. Since all strings in the set are distinct, they must differ in at least one position (without loss of generality, in the first bit). Divide the set into two groups such that all strings in the first group start with a 0 and all strings in the second start with a 1. Notice that any two strings in differing groups are distinguished by the first index, but any two strings in the same group must still all be distinct when restricted to the last $n-1$ bits. If the first group is of size $k$ and second is of size $m-k$, then by induction we can distinguish the strings within the group with at most $(k-1) + (m-k-1) = m - 2$ indices. Adding the index we used for the first comparison, we arrive at the conclusion.
\end{proof}

\noindent{\bf Proof of Proposition~\ref{prop:wide}}
Suppose that we have $m > 1$ distinct clauses.   The strategy is as follows:  for each subset of $m-1$ variables, query the oracle with all possible 0-1 assignments to those variables, setting each variable not in the set to $1/2$.  Using Lemma~\ref{lem:wide}, we see that we will eventually query the subset of the variables which causes the oracle to return all $m$ clauses.  Therefore, we learn the presence of $m-1$ variables exactly as they appear in the hidden formula.  Once we have this information, it is a simple task to recreate each clause $C_k$.  Set the $m-1$ variables such that none satisfy clause $C_k$.  Notice that all other clauses will be satisfied, so we have reduced this to our problem of $\WIDE$ with one clause, which we know how to learn in $O(n)$ time. $\Box$

\section{Hidden 1SAT}
\label{sec:1SAT}
%Adam will write this first
In this section, we will consider the problem of a hidden $\SAT{1}$ instance $\Phi$, possibly with repetitions. 
Our goal will be to determine whether or not $\Phi$ is satisfiable. A natural approach one might take to solve this problem would be to learn the identity of each clause in the instance $\Phi$. Unfortunately, in the case that the $\SAT{1}$ instance has repetitions, this is not possible.
\begin{prop} 
\label{clm:nolearn}
There is no algorithm which, given an instance $\Phi$ which is unsatisfiable, learns all the literals present in $\Phi$ (even granted arbitrary numbers of queries to the oracle).
\end{prop}
\begin{proof}
Consider the following two $\HSAT{1}$ instances:
\begin{align*}
\Phi_1: C_1 = x_1, C_2=\overline{x}_1, C_3=x_1, C_4=\overline{x}_1 \\
\Phi_2: C_1=x_1, C_2 = \overline{x}_1, C_3=x_2, C_4=x_2
\end{align*}
Both of these instances are unsatisfiable. However, note that for any oracle query, it is possible for the oracle to give the same answer (i.e. clause index) for each query. 
To see this, if $x_1$ is more violated than $\overline{x}_1$ or $x_2$, return clause $C_1$. 
If $\overline{x}_1$ is more violated than $x_1$ or $x_2$, then return clause $C_2$. 
If $x_2$ is more violated than $x_1$ or $\overline{x}_1$, then return clause $C_3$ if $x_1$ is more violated than $\overline{x}_1$, otherwise return clause $C_4$. One can easily check these oracle answers are consistent with either instance. Hence these instances are indistinguishable to adversarial oracle answers, so no algorithm can distinguish $\Phi_1$ and $\Phi_2$.
\end{proof}

Here the difficulty in learning an unsatisfiable instance does not lie in the repetition of clauses, but rather in determining for which $i $ do both $x_i$ and $\xbar_i$ appear in $\Phi$. This shows that no algorithm can learn the hidden 1SAT instance \footnote{Note, however, it is still possible that there exists an algorithm to learn the $\SAT{1}$ instance when the instance is promised to be satisfiable.}. Hence if there is an algorithm to solve $\SAT{1}$ in this hidden setting, then it must solve the instance despite the fact that it cannot deduce the underlying instance. Surprisingly, this turns out to be possible.
\begin{theorem} 
\label{thm:1SAT}
Given a hidden $\SAT{1}$ instance $\Phi$ on $n$ variables and $m$ clauses, it is possible to determine if $\Phi$ is satisfiable in time $O(mn^2)$.
\end{theorem}
%\vspace{-0.5em}
\begin{proof}
Consider an ordering of the variables $x_1...x_n$. The algorithm
will work by inductively constructing a series of lists
$L_1,L_2,\ldots L_{n}$. Each list $L_i$ will contain a list of
partial assignments to the variables $x_1\ldots x_i$. Each list will
be of size at most $m$, with the exception of $L_n$ which will be of
size at most $2m$. Let us call a partial assignment $p$ to
$x_1\ldots x_i$ \emph{good} if there exists an assignment $p'$ to
the variables $x_{i+1}\ldots x_n$ such that the assignment $p\cup
p'$ satisfies $\Phi$. Correspondingly, call $p$ \emph{bad} if it
cannot be extended to a satisfying assignment of $\Phi$. (Note in
the case of 1SAT, every partial assignment is either good or bad.)
Our algorithm will guarantee that, if $\Phi$ is satisfiable, then at
least one assignment in each list is ``good''. Therefore, by
constructing the list $L_{n}$, then trying all assignments in
$L_{n}$, we will be guaranteed to find a satisfying assignment if
one exists.

We now describe how to construct the lists $\{L_i\}_{i\in[n-1]}$ by induction. 
The base case of $L_1$ is trivial - just add both $x_1=0$ and $x_1=1$ to the list. 
We now show how to construct $L_{i+1}$ given $L_i$. 
First, let $\tilde{L}_{i+1}$ be all possible extensions of the assignments in $L_i$ to the variable $x_{i+1}$. 
Clearly if one of the assignments in $L_i$ was good, then one of the assignments in $\tilde{L}_{i+1}$ is good.
However, when $i+1<n$, the size of $\tilde{L}_{i+1}$ could become too large - it is of size $2|L_i|$ which could at some point become larger than $m$. So we need to reduce the size of $\tilde{L}_{i+1}$ so that it contains at most $m$ partial assignments. To decide which partial assignments to keep, we will perform the following oracle queries: for each partial assignment $y\in\tilde{L}_{i+1}$, propose the following query $q_y$ to the oracle: set $x_1...x_{i+1}$ to 0 or 1 according to $y$, and set all other variables to value $1/2$. 
The oracle will return the identity of a clause $C_j$ which is worst violated by this fractional assignment. 
Now partition the elements of $\tilde{L}_{i+1}$ according to which clause $C_j$ was returned by the query.
This divides the elements of $\tilde{L}_{i+1}$ into at most $m$ equivalence classes.
To construct $L_{i+1}$, simply pick one element from each equivalence class of $\tilde{L}_{i+1}$.

Clearly $L_{i+1}$ has size at most $m$ by construction.
To complete the proof, we need to show that at least one element of $L_{i+1}$ is good.
First, by the induction hypothesis, at least one element of $L_{i}$ is good. 
This implies at least one element  $y^*\in\tilde{L}_{i+1}$ is good as well.
Consider what happens when we perform the query $q_{y^*}$.
Since $y^*$ is good, $q_{y^*}$ must satisfy all clauses involving the variables $x_1\ldots x_{i+1}$.
If there are no clauses involving the remaining variables $x_{i+2}\ldots x_n$, then $q_{y^*}$ satisfies the instance, so the oracle will tell us this and we can terminate the algorithm.
Otherwise, there is a clause involving some variable in $\{x_{i+2}\ldots x_n\}$. 
When we query $q_{y^*}$, the worst violated clause will be some clause $C_k$ involving a variable in $\{x_{i+2}\ldots x_n\}$, which will be violated with probability $1/2$. 
So the equivalence class corresponding to $C_k$ will contain a good assignment.
Furthermore, since $C_k$ involves one of the variables in $\{x_{i+2}\ldots x_n\}$, it will never be returned as the worst violated clause for query $q_{y'}$ for any bad assignment $y'\in\tilde{L}_{i+1}$, because any bad assignment will violate a clause involving $\{x_1\ldots x_{i+1}\}$ by 1, while $C_k$ will be violated only with probability 1/2.
Therefore the equivalence class corresponding to $C_k$ will contain only good assignments.
So by picking one assignment from each equivalence class, we will ensure $L_{i+1}$ contains at least one good assignment, as claimed.

The time to construct each list is $O(mn)$, and the algorithm constructs $n$ lists. Hence the algorithm runs in time $O(mn^2)$.
\end{proof}

%\vspace{-1.5em}
\section{Hidden 2SAT without repetitions}
\label{sec:2SAT}
%\vspace{-0.5em}
In this section, we consider a hidden $\SAT{2}$ formula $\Phi$ which is promised to contain no two clauses that are the same.  Although Proposition~\ref{clm:nolearn} shows that we cannot always hope to learn $\Phi$ directly, it does not rule out the possibility of learning some $\Phi'$ such that $\sat(\Phi') = \sat(\Phi)$.  In fact, this is exactly the approach we take.  

\begin{theorem} 
\label{thm:2SAT_learn}
Suppose $\Phi$ is a repetition-free $\HSAT{2}$ instance on $n$ variables.  Then, we can generate $\Phi'$ such that $\sat(\Phi') = \sat(\Phi)$ in $O(\poly(n))$.
\end{theorem}

Before proving Theorem~\ref{thm:2SAT_learn}, we show that a satisfying assignment can be found when 
$\Phi$ is a satisfiable repetition-free $\SAT{2}$ instance. 

\begin{theorem}
\label{thm:rep_free_2SAT}
Suppose $\Phi$ is a hidden repetition-free $\SAT{2}$ instance on $n$ variables.  Then it is possible to generate a satisfying assignment in time $O(n^2)$.
\end{theorem}
%\vspace{-0.5em}
\begin{proof}
The idea is to attempt to learn each clause present in the formula.  Suppose we wish to determine if the clause $(x_i \vee x_j)$ is present in $\Phi$ (an analogous procedure works to determine if a $\SAT{1}$ clause $x_i$ is in $\Phi$).  We can assume that the clause is unobscured because the presence of an obscured clause does not affect the set of satisfying assignments.  Run the following procedure:
\begin{enumerate}
\item First query the oracle with the assignment $x_i =0$, $x_j =0$, $x_k =0$ for $k \neq i, j$.  If this is a satisfying assignment, then we are done.  Otherwise, we know that there must exist a clause of type:
\begin{inparaenum}[\upshape(\itshape a\upshape)]
\item $(x_i \vee x_j)$;
\item $(x_i \vee x_k)$ for $k \neq i, j$;
\item $(x_j \vee x_k)$ for $k \neq i, j$; or
\item $(x_{k_1} \vee x_{k_2})$ for $k_1, k_2 \neq i, j$.
\end{inparaenum}
\item Now query the oracle with the assignment $x_i =0$, $x_j =0$, $x_k =1$ for $k \neq i, j$.  As before, if this is satisfying, we are done.  Otherwise, we know that there must exist a clause of type:
\begin{inparaenum}[\upshape(\itshape a\upshape)]
\item $(x_i \vee x_j)$;
\item $(x_i \vee \bx_k)$ for $k \neq i, j$;
\item $(x_j \vee \bx_k)$ for $k \neq i, j$; or
\item $(\bx_{k_1} \vee \bx_{k_2})$ for $k_1, k_2 \neq i, j$.
\end{inparaenum}
\item We can now construct an explicit test for the presence of the clause $(x_i \vee x_j)$.  We will propose two fractional assignments to the oracle.  If $(x_i \vee x_j)$ is present, then the clause returned each time will be the same.  If it is not present, then the returned clause will be different.  Formally, query the oracle with the assignment $x_i =0$, $x_j =0$, $x_k =\frac{1}{4}$ for $k \neq i, j$ and then with the assignment  $x_i =0$, $x_j =0$, $x_k =\frac{3}{4}$ for $k \neq i, j$.  Table~\ref{table:davis_put} shows the accompanying violations.
\begin{table}[ht!]
\centering
\scalebox{.9}{
\begin{tabular}{c | c | c | c | c | c | c | c | c}
& $(x_i \vee x_j)$ & $(x_i \vee x_k)$  & $(x_j \vee x_k)$ & $(x_{k_1} \vee x_{k_2})$ & $(x_{k_1} \vee \bx_{k_2})$ & $(x_i \vee \bx_k)$ & $(x_j \vee \bx_k)$ &  $(\bx_{k_1} \vee \bx_{k_2})$ \\ \hline
1/4  & 1 & 3/4 & 3/4 & 9/16 & 3/16 & 1/4 & 1/4 & 1/16 \\
3/4  & 1 & 1/4 & 1/4 & 1/16 & 3/16 & 3/4 & 3/4 &  9/16
\end{tabular}}
\caption{Violation of the clauses based on the fractional assignments of 1/4 and 3/4.}
\label{table:davis_put}
\end{table}
It is clear that if $(x_i \vee x_j)$ is present in the formula, then it is returned on both assignments.  If it is not present, then from the table we can also see that one of the clauses known to exist from our first query must be returned on the 1/4 fractional assignment.  However, one of the clauses known to exist from our second query must be returned on the 3/4 fractional assignment.  Thus, the clause returned by the oracle changes when $(x_i \vee x_j)$ is not present.
\end{enumerate}
Notice that the above procedure also works to detect all $\SAT{1}$ and $\SAT{2}$ clause types.  Therefore, if we complete the above procedure with all $O(n^2)$ clause types without finding a satisfying assignment, then we have identified all  unobscured clauses in the formula. It is clear that the conjunction of these clauses forms a formula $\Phi'$ such that $\sat(\Phi') = \sat(\Phi)$.  Therefore, we can use any 2SAT algorithm which runs in time $O(n^2)$ on $\Phi'$ to find some satisfying assignment of $\Phi$.
\end{proof}

\noindent{\bf Proof of Theorem~\ref{thm:2SAT_learn}}
First run the procedure in the proof of Theorem~\ref{thm:rep_free_2SAT}.  Notice that it either learns all unobscured clauses or produces a satisfying assignment.  Therefore, let us assume that we have some satisfying assignment, and without loss of generality, let us assume that it is the all ones assignment.  This implies that each clause contains at least one positive literal.

Let us now give a procedure to find all those variables which must
necessarily be set to 1; the $\SAT{1}$ clauses corresponding to
these variables will be added to $\Phi'$.  To determine whether or
not $x_i$ must be set to 1 in $\Phi$, set all queries in
Theorem~\ref{thm:rep_free_2SAT} so that $x_i = 0$.  If the formula
$\Phi$ is unsatisfiable, then we know that $x_i$ must be set to 1. 
If, however, all such variable assignments are still satisfiable,
then all remaining clauses that are unobscured by the current
assignment must be clauses on 2 variables.  Furthermore, since
all-ones is a satisfying assignment, these clauses must either be of
the form $(x_i \vee x_j)$ or $(x_i \vee \bx_j)$.  Clearly, the above
procedure ends after a polynomial number of steps.  Notice that this
corresponds to reaching a ``branch point'' in a typical $\SAT{2}$
algorithm.

To simplify the exposition, suppose that after the above procedure
the remaining variables are still on $x_1, \ldots, x_n$ and that we
are querying the formula $\Phi$.  Let us first attempt to learn the
clauses of the form $(x_i \vee \bx_j)$.  Once again, set $x_i=0$ and
$x_j = 1$ and use Theorem~\ref{thm:rep_free_2SAT} to determine if
the induced formula is satisfiable.  If it is satisfiable, then
clearly $\Phi$ did not contain $(x_i \vee \bx_j)$.  However, if it
isn't satisfiable, then indeed $(x_i \vee \bx_j)$ must be present in
$\Phi$ because the assignment $x_i=0$, $x_k =1$ for all $k \neq i$
satisfies all other $\SAT{2}$ clauses.

We now only left to learn the clause of the form $(x_i \vee x_j)$. 
Suppose we propose the assignment $x_i = 0, x_j = 0,$ and $x_k = 1$
for all $k \neq i, j$.  If the formula is not satisfied then the
returned clause type must be one of
\begin{inparaenum}[\upshape(\itshape a\upshape)]
\item $(x_i \vee x_j)$;
\item $(x_i \vee \bx_k)$; or
\item $(x_j \vee \bx_k)$;
\end{inparaenum}
as all such clauses are maximally violated (with probability
1).  Now propose the assignment $x_i = 0, x_j = 0,$ and $x_k = .5$
for all $k \neq i, j$.  Clearly, if $(x_i \vee x_j)$ is not present,
the oracle returns a clause that contains both a positive and
negative literal.  Having already learned clauses of this type, we
know that $(x_i \vee x_j)$ must not be present in $\Phi$.  Likewise,
if $(x_i \vee x_j)$ is present, the oracle returns a clause index
which we have not previously learned, from which we can conclude
that $(x_i \vee x_j)$ is indeed present. $\Box$\\

%\vspace{-0.5em}
While the above procedure may seem elementary, it acts as a stepping stone to tackle the harder problem of learning an unknown input instance of quantum $\SAT{2}$, which is introduced and discussed in the subsequent sections. 

%\vspace{-1.5em}
\section{Quantum SAT Preliminaries}
%\subsection{Notation}
\vspace*{-0.5em}
%\anote{Aarthi}{Add notations for states, density operators and Bloch Sphere that we use}
\subparagraph*{Notations.} A quantum system of $n$ qubits is
described using a Hilbert space $\cc{H} = \cc{H}_1 \otimes \cc{H}_2
\otimes \ldots \otimes \cc{H}_n$ where each $\cc{H}_i$ is a
two-dimensional Hilbert space of the $i^{th}$ qubit. Vectors in
$\cc{H}$ are called \emph{pure states} and they describe a state of
the system. By adding a subscript $i$ to
the vector $\ket{\alpha}{}$ we indicate that $\ket{\alpha}{i}$ is
defined in the local Hilbert space $\cc{H}_i$ of the $i^{th}$ qubit.
Similarly, $\ket{\psi}{ij}$ denotes a $2$-qubit state $\ket{\psi}{}$
in $\cc{H}_i \otimes \cc{H}_j$. In any local qubit space $\cc{H}_i$,
we pick an orthonormal basis $\ket{0}{}, \ket{1}{}$ so that every 
$1$-qubit state $\ket{\alpha}{}$ can be expanded as $\ket{\alpha}{}=
\alpha_0 \ket{0}{} + \alpha_1 \ket{1}{}$.  We define its orthogonal
state by $\ket{\alpha^\bot}{}\EqDef \alpha_1 \ket{0}{} - \alpha_0
\ket{1}{}$;\footnote{There are, of course, continuously many
orthogonal states for every $\ket{\alpha}{}$, so here we simply
choose one in a canonical way.} clearly, $\qip{\alpha}{\alpha^\bot}
= 0$. A standard geometrical representation of the state space of a
single qubit is the Bloch sphere which is illustrated, for
completeness, in Figure~\ref{fig:Bloch}.

\begin{center}
\begin{figure}[h!]
	\begin{minipage}[b]{0.5\textwidth}
	The Bloch sphere is a geometrical representation of the state space of a single qubit and as the name suggests is a sphere~\cite{NC00}. The antipodal points of the Bloch sphere corresponds to orthogonal states and in general, the north and south poles are usually indexed as $\ket{0}{}$ and $\ket{1}{}$ respectively as shown in Figure~\ref{fig:Bloch}. The points on the surface of the sphere correspond to pure states and the interior of the sphere corresponds to mixed states - probabilistic mixtures of pure states. The center of the sphere is the completely mixed state $(\mathbb{I}/2)$ which can be interpreted as an equal parts mixture of any state and its orthogonal state. The interested reader is referred to~\cite{NC00} for more details on the exact correspondence between quantum states and points on the Bloch Sphere.
 	\end{minipage}
 	\hfill
 	\begin{minipage}[b]{0.4\textwidth}
  	\centering
  	\def\svgwidth{150pt}
  	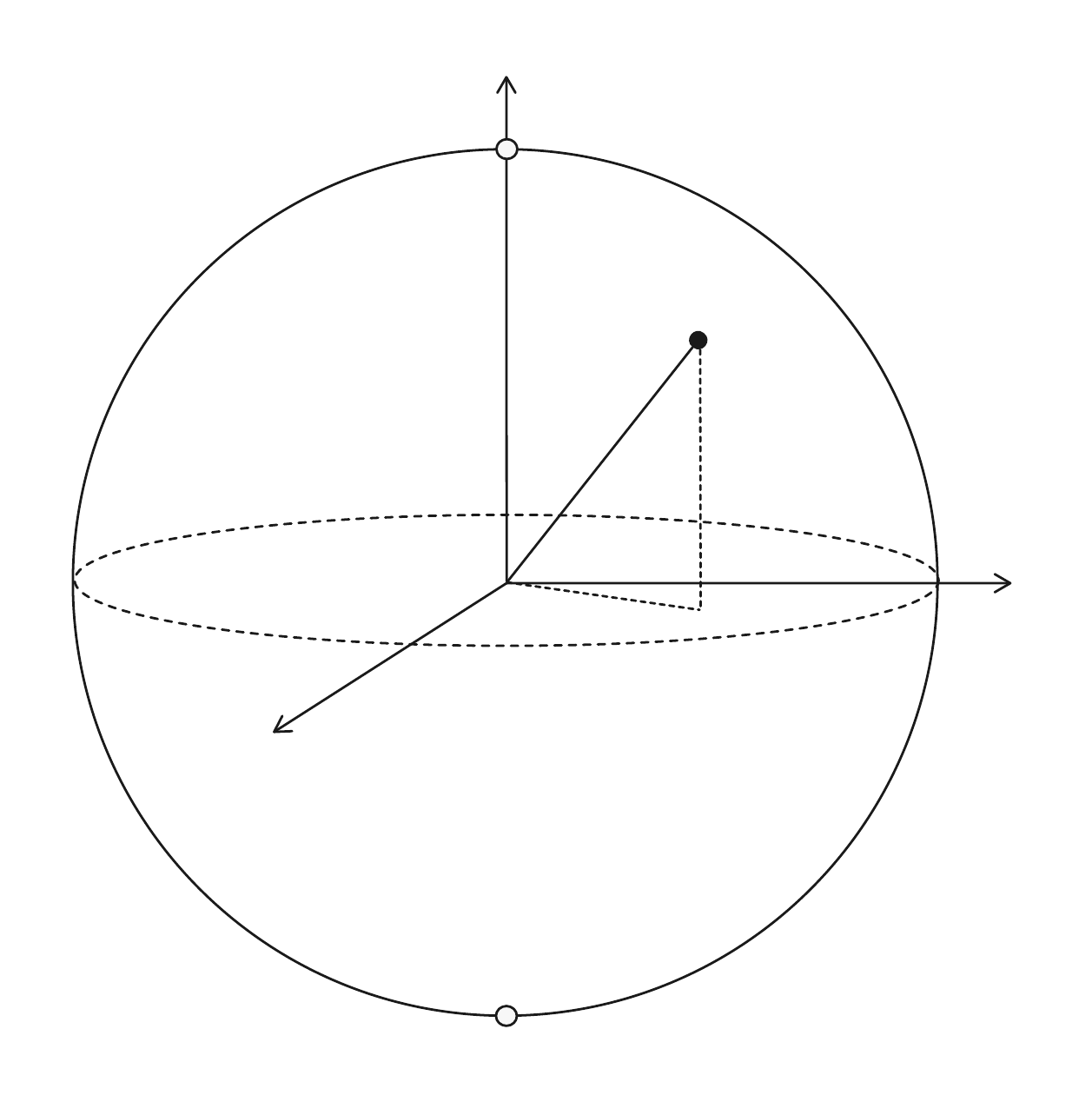
  	\caption{The Bloch sphere with $\ket{\psi}{}$}
  	\label{fig:Bloch}
	\end{minipage}
\end{figure}
\end{center}  

\vspace*{-1.5em}
%The explicit correspondence between elements of the Bloch sphere and single qubit quantum states can be found in any
A more general way to describe a quantum state is by its
\emph{density matrix}.  Density matrices can be
viewed as statistical ensembles of pure states that are described by
vectors. A density matrix representation a single pure state
$\ket{\psi}{}$ is given by the matrix $\rho = \bk{\psi}{}$. General
density matrices are given as a convex sum of density matrices of
the pure states with the coefficient summing up to $1$: $\sigma =
\sum_i p_i \bk{\psi_i}{}$ where $\forall i, p_i \geq 0$ and $\sum_i
p_i = 1$. Alternatively, they are defined as semi-definite operators
whose trace is equal to $1$. For instance, the density matrix
$\frac{1}{2}\mathbb{I}$ can be written as $\frac{1}{2}\mathbb{I} =
\frac{1}{2}\bk{0}{} + \frac{1}{2}\bk{1}{}$. The state of a quantum
system can always be fully specified by a density matrix.

Observables in quantum mechanics are associated with Hermitian
operators. The eigenvalues of such an operator correspond to the
possible outcomes of a measurement.  Given such a Hermitian
operator $A$ and a pure state $\ket{\psi}{}$, the expression
$\bra{\psi}{}A\ket{\psi}{}$ is the \emph{expectation value} of $A$.
It is the result we get if we measure $A$ over many copies of the
same state $\ket{\psi}{}$ and average the result. One can use the
Chernoff bound to deduce that, with high probability, if we measure
$A$ over $\poly(n)$ copies of a state $\ket{\psi}{}$, we obtain an
approximation to $\bra{\psi}{}A\ket{\psi}{}$ with an additive error
of $1/\poly(n)$. 

The expectation value of $A$ with respect to a state which is
described by a density matrix $\rho$ is given as $\Tr(\rho A)$. Note
that if $\rho$ is given by $\rho=\sum_i p_i
\ket{\psi_i}{}\bra{\psi_i}{}$ with $\sum_i p_i=1$, then $\Tr(\rho A)
= \sum_i p_i\bra{\psi_i}{}A\ket{\psi_i}{}$, which justifies the
interpretation of $\rho$ as a statistical ensemble of pure states.
Like in the pure state case, using $\poly(n)$ identical copies of
$\rho$, one can estimate the expectation value $\Tr(\rho A)$ up to
an additive error of $1/\poly(n)$.
 
%\subsection{Local Hamiltonians and Quantum SAT}
%\vspace*{-2mm}
\vspace*{-0.5em} 
\subparagraph*{Local Hamiltonians and Quantum
SAT.} While classically $\SAT{}$ is given as a CSP, quantum
$\SAT{k}$ ($\QSAT{k}$) is defined as a special case of the $k$-local
Hamiltonian problem. A $k$-local Hamiltonian on $n$ qubits is a
Hermitian operator $H = \sum_{e=1}^m h_e$, where each $h_e$ is a
\emph{local} Hermitian operator acting non-trivially on at most $k$
qubits. Formally, it is written as $h_e =
\hat{h}_e\otimes\mathbb{I}_{rest}$, where $\hat{h}_e$ is defined on
the Hilbert space of $k$ qubits, and $\mathbb{I}_{rest}$ is the
identity operator on the Hilbert space of the rest of the qubits.
When it is clear from the context, we often use $h_e$
instead of $\hat{h}_e$, even while referring to its action on the
local Hilbert space.

In physics, $k$-local Hamiltonians model the local interactions
between particles in a many-body system and are the central tool for
describing the physics of such systems. The \emph{energy} of the
system for every state $\ket{\psi}{}$ is defined by $E_\psi(H)
\EqDef \bra{\psi}{}H\ket{\psi}{} = \sum_e
\bra{\psi}{}h_e\ket{\psi}{}$. The lowest possible energy of the
system is called the \emph{ground energy} and is denoted by
$E_0(H)$. It is easy to verify that $E_0(H)$ is
the lowest eigenvalue of $H$. The corresponding
eigenspace is called the \emph{ground space} of the system, and
its eigenvectors are called \emph{ground states}. A
central task in condensed matter physics is to understand the
properties of the ground space, as it determines the low-temperature
physics of the system.

%\anote{Itai}{Throught this section, I replaced the notation
%$H'=\sum_e \Pi'_e$ with $H=\sum_e \Pi_e$, which I find much more
%sensible (i.e., the exact system is without the prime and the
%approximate system is primed). But if you don't agree, change it back}
There is a deep connection between the problem of approximating the
ground energy of a local Hamiltonian and the classical problem of
finding an assignment with minimal violations in a local CSP. In
both cases, one tries to minimize a global function that is given in
terms of local constraints. This connection is evident if we
consider the special case when the local Hermitian operators $h_e$
are given as local \emph{projectors} $\Pi_e$. Then for any state
$\ket{\psi}{}$, the local energy $\bra{\psi}{}\Pi_e\ket{\psi}{}$ is
a number between 0 and 1 that can be viewed as a measure to how much
the state is `violating' the quantum clause $\Pi_e$. When the local
energy is $0$, the state is inside the null space of the projector
$\Pi_e$ and is said to satisfy the constraint. The total energy of
the system, $E_\psi=\bra{\psi}{}H\ket{\psi}{} = \sum_e
\bra{\psi}{}\Pi_e\ket{\psi}{}$ then corresponds to the total
violation of the state $\ket{\psi}{}$. When the ground energy of the
system is $0$, necessarily the ground space is the non-vanishing
intersection of all the null spaces of the local projectors, and we
say that the system is satisfiable. From a physical point of view,
such a system is called \emph{frustration-free}, since any ground
state of the global system also minimizes the energy of every local
term $\Pi_e$.

The quantum $\QSAT{k}$ problem is analogous to the classical
$\SAT{k}$ problem. Whereas in the $\SAT{k}$ case we are asked to
decide whether a $k$-local CSP is satisfiable or not, in the
$\QSAT{k}$ problem we are asked to determine whether the ground
energy of a $k$-local Hamiltonian made of projectors is 0 or not.
Unlike the truth values of $\SAT{}$ clauses, however, the ground
energy of a $k$-local Hamiltonian is a continuous function that is
sensitive to any infinitesimal change in the form of the local
projectors. To make the $\QSAT{k}$ problem more physically relevant,
we define it using a promise: Given a $k$-local Hamiltonian of
projectors over $n$ qubits and a value $b > \frac{1}{n^\alpha}$ for
some constant $\alpha$, decide if the ground energy of $H$ is $0$
(the $\yes$ case) or the ground energy of $H$ is at least $b$ (the
$\no$ case). Bravyi~\cite{Bra06} showed that $\QSAT{k}$ for $k \geq
4$ is $\QMA_1$-complete while Gosset and Nagaj~\cite{GN13} showed
that $\QSAT{3}$ is also $\QMA_1$-complete. The class $\QMA_1$ stands
for `Quantum Merlin Arthur' with one-way error, and is the quantum
generalization of the classical $\MA_1$ class with one-way error.
The differences are that the witness can be a quantum state over
$\poly(n)$ qubits, and the verifier can be an efficient quantum
machine. In Ref.~\cite{Bra06} it was known that $\QSAT{2}$ has an
$O(n^4)$ classical algorithm, and is therefore in $\P$. More
recently linear time algorithms for the same problem have been
constructed~\cite{ASSZ15,dBG15}. 

As the Hamiltonian in a $\QSAT{2}$ instance is a sum of $2$-qubit
projectors, every local projector is defined on a $4$-dimensional
Hilbert space and is of rank $1, 2$ or $3$. The non-zero subspace 
of each projector (the subspace on which it projects) is commonly
referred to as the \emph{forbidden space} of that projector and the
orthogonal subspace is its \emph{solution space}. Finally, we say
that $H$ has \emph{no repetitions} if there does not exist any pair
of different projectors $\Pi_{e},\Pi_{e'}$ which act non-trivially
on the same set of qubits. In the case of repetition free
$\QSAT{2}$, each projector can also be indexed by the qubit pairs it
acts on and the instance can be written as $H = \sum_{(u,v) \in S}
\Pi_{uv}$, where $S\subseteq [n]\times[n]$ and each $\Pi_{uv}$ is
non-zero. For any projector $\Pi$ and a state $\ket{\psi}{}$, we say
that $\ket{\psi}{}$ \emph{satisfies} $\Pi$ up to $\epsilon$ if
${E_\psi (\Pi)}\EqDef \bra{\psi}{}\Pi\ket{\psi}{} \leq \epsilon^2$.
The energy $E_\psi (\Pi)$ is the \emph{violation energy} of
$\ket{\psi}{}$ with respect to the projector $\Pi$. Notice that when
the state of the system is described by a density matrix $\rho$, its
violation energy with respect to $\Pi$ is given by
$E_\rho(\Pi)\EqDef \Tr(\rho\Pi)$
 
Finally, a $\QSAT{2}$ Hamiltonian $H$ is said to have a \emph{Star-like}
configuration if there exists a pair of qubits $u,v$ with
$\Pi_{u,v}\neq 0$ such that \emph{all} projectors involve either $u$
or $v$.
 
%This notion becomes clearer if one constructs an interaction graph for $H$ where each qubit is a vertex and there is an edge between $i$ and $j$ if $\Pi_{ij} \neq 0$ as shown in Figure~\ref{fig:star}.\anote{Aarthi}{Add figure of star like configurations.}
%The same terminology also holds for the case of $\QSAT{1}$ which is defined in (ADD THIS IN Q1SAT SECTION)
%For any projector $\Pi$ and a state $\ket{\psi}{}$, we say that $\ket{\psi}{}$ satisfies $\Pi$ up to $\epsilon$ if ${E_\psi (1 - \Pi)} \leq \epsilon^2$.

%\subsection{Hidden QSAT}
%\label{sec:hqsat}
\vspace{-0.5em}
\subparagraph*{Hidden QSAT.} The hidden version
of $\QSAT{}$ is defined analogously to the classical case. Our task
is to decide whether a $k$-local Hamiltonian $H=\sum_e \Pi_e$ that
is made of $m$ $k$-local projectors over $n$ qubits is
frustration-free with $E_0=0$ ($\yes$ instance) or $E_0 > m \cdot
2\epsilon^2$ ($\no$ instance). Here, $\epsilon>0$ is some threshold
parameter that can be assumed to be inverse polynomially small in
$n$. Moreover, as in $\HSAT{}$, here we do not
know the Hamiltonian itself; instead we can only send \emph{quantum}
states to a ``worst violated oracle'', which will return the
index $e$ of the projector $\Pi_e$ with the highest violation
energy. Since we want to generalize the notion of a probabilistic
assignment that is used in $\HSAT{}$, we allow
ourselves to send the oracle qubits that hold a general quantum
state $\rho$, which can only be described by a density matrix. 
Recall from the previous section that this can be regarded as an
ensemble of pure quantum states. Then the oracle will return the
the index $e$ for which $\Tr(\Pi_e\rho)$ is maximized. If the total energy of the proposed state is $\leq m \cdot \epsilon^2$ then the oracle will indicate that a satisfying assignment has been found.

\section{Hidden Quantum 1SAT} \label{sec:1qsat} 
The algorithm used to solve $\HSAT{1}$ can be extended to solve the
$\HQSAT{1}$ problem as well. A $1$-local projector defined on
$\mathbb{C}^2$ is satisfiable if it is of rank at most $1$ and can
be viewed as setting the direction of the qubit on the Bloch sphere.
Unlike the classical case, where we may view the $\SAT{1}$ clauses as
either the $\bk{0}{}$ or $\bk{1}{}$ projectors, here the projectors can
point in any direction in the Bloch sphere. To handle the continuous 
nature of the Bloch Sphere, we consider discretizing it by using an 
$\epsilon$-net that covers the whole sphere. This allows us to 
generalize the lists of $0-1$ strings used in $\HSAT{1}$ into lists 
of $n$-qubit product states where each qubit is assigned an element 
of the $\epsilon$-net.

Given a 1-local projector $\bk{\psi}{}$, its zero space is spanned by
$\ket{\psi^{\bot}}{}$. We can divide the Bloch sphere into two hemispheres,
one hemisphere containing states $\ket{\phi}{}$ having $| \qip{\psi}{\phi}| \leq \frac{1}{2}$ 
and the other with states having $ | \qip{\psi}{\phi}| > \frac{1}{2}$. An $n$-qubit state 
$a = \ket{a_1}{} \ket{a_2}{} \ldots \ket{a_n}{}$ is called \emph{good} if for each qubit $i$, 
where $\ket{\psi_i}{}$ is its forbidden state, $|\qip{\psi_i}{a_i}| \leq \frac{1}{2}$ and 
\emph{bad} if $\forall i, |\qip{\psi_i}{a_i}| > \frac{1}{2}$. 
For the $n$-qubit state $a = \ket{a_1}{} \ket{a_2}{} \ldots \ket{a_n}{}$, let
$a' \EqDef \ket{a_1^\bot}{} \ket{a_2^\bot}{} \ldots \ket{a_n^\bot}{}$.

Now, we can sketch the $\HQSAT{1}$ algorithm. Adapting the 
process described in Theorem~\ref{thm:1SAT} for an arbitrary $n$-qubit 
state $a$ gives a list of $n$-qubit states, $L_{a / a'}$, 
where at least one state is \emph{good}. This is formally stated in 
Lemma~\ref{lem:Q1SATa}. 

\begin{lemma}
\label{lem:Q1SATa} 
  Let $a=\ket{a_1}{} \otimes \ldots \otimes \ket{a_n}{}$ where $\ket{a_i}{}, \ket{a_i^\bot}{}$ is a basis for
  qubits $i$, for $i = 1, \ldots, n$. Then one can produce a list,
  $L_{a / a'} \subset \bigotimes_{i=1}^n \{\ket{a_i}{},
  \ket{a_i^\bot}{}\}$ of at most $2mn$ states such that, if the
  instance is satisfiable, there is at least one \emph{good}
  $n$-qubit state in the list. The time taken to produce this list
  is $O(n^2m)$.
\end{lemma}

\begin{proof}
Construct the list inductively as in Theorem~\ref{thm:1SAT} where at stage $k$, $L_{k, a / a^{\bot}}$ contains at most $m$ strings at least one of which is good for qubits $1, \ldots, k$ if the instance is satisfiable, by constructing trials as follows. At stage $k$, replace $\{0, 1\}$ with $\{a_k, a_k^\bot\}$ for qubit $k$ and the value of $\frac{1}{2}$ with $\frac{\mathbb{I}}{2}$, the completely mixed state, for qubits $k+1, \ldots, n$ so that $L_{k, a / a^{\bot}}$ contains states from the set 
\begin{align*}
\bigotimes_{i=1}^k \{\ket{a_i}{}, \ket{a_i^\bot}{}\} \otimes \left(\frac{\mathbb{I}}{2}\right)^{\otimes (n-k)}.
\end{align*}
This almost finishes the process except for one caveat when there exists no projector on qubits $k+2, \ldots, n$ while constructing $L_{k+1, a / a^{\bot}}$ from $L_{k, a / a^{\bot}}$. This situation also occurs at the last step while constructing $L_n$. In both cases, while proposing a good state, all violations are $\leq \frac{1}{2}$ and any clause $\id$ returned by the oracle involves a qubit in $1, \ldots, k+1$. This same clause could also be violated with probability $> \frac{1}{2}$ when a bad string is proposed which will incorrectly be put in the same equivalence class as the good one. Then, picking just one representative from $C_j$ is insufficient and the size of the lists cannot be compressed. To fix this, we add the following checks:
\begin{enumerate}
\item If $k+1 = n$, just double the number of strings on the list, assuming that there is a clause involving $n$, i.e. the last qubit that is assigned values.
\item Repeat the algorithm $n$ times by placing a different qubit at the last position each time. Let $L_{a / a^{\bot}}$ be the union of all the lists found in this manner and is of size $2m \cdot n$. For a non-empty instance, at least one of the trials in the list will be such that there is a clause on the qubit in the last position and will hence contain a good string.
\end{enumerate}
Hence, $L_{a / a^{\bot}}$ with $2mn$ strings contains at least one good string and by repeating the classical process $n$ times, we get an $O(mn^2)$ time procedure for this.
\end{proof}

\noindent However, this only gives us an assignment that violates 
each projector by $\leq \frac{1}{4}$ while we require assignments 
that violate each projector by $\leq \epsilon^2$. The key observation involves 
constructing two lists $L_{a / a'}$ and $L_{b / b'}$ 
where $b \neq a, a'$ and picking a state from each
list. Consider the case when both states are good. Let the states
on qubit $i$ from each list be $\ket{a_i}{}$ and $\ket{b_i}{}$
respectively. Each state defines a hemisphere $R_{i, a_i}$ and
$R_{i, b_i}$ containing all the states that are bad with respect to 
the forbidden state for qubit $i$, $\ket{\psi_i}{}$. Then, 
$\ket{\psi_i}{}$, should be contained in 
$R_{i,a_ib_i} \EqDef R_{i, a_i} \cap R_{i, b_i}$. The optimal choice for
$b_i$, given $a_i$, would be one where $|R_{i, a_ib_i}| \leq
\frac{|R_{i, a_i}|}{2}$. Then, similar to performing a binary search
on the Bloch Sphere, repeating this process $\log_2
\left(\frac{1}{\epsilon}\right)$ times, will give a region
consisting of good approximations to the forbidden state (See
Figures~\ref{fig:Q1SAT} $(a)$ and $(b)$ for illustrations).

\begin{theorem}
\label{thm:Q1SATb} 
  Let $\epsilon > 0$. Given a $\HQSAT{1}$ on $n$ qubits containing $m$ projectors, there exists an an $O((2mn)^{2 \log \frac{1}{\epsilon}} \cdot mn^2)$ time algorithm, with the property that \vspace{-0.5em}
  \begin{enumerate}[(a)]
  \item for a frustration free instance, it outputs an assignment where for each projector, the
  forbidden state is violated with probability $\leq \epsilon^2$ and
  \item for a $\no$ instance, the algorithm outputs $\unsat$.
  \end{enumerate}  
\end{theorem}
\begin{proof}
  Initially, with no information, for each qubit $i$, $R_i =$ Bloch
  sphere. Now the algorithm executes the following steps:
\begin{itemize}
\item \vspace*{-0.5em} Start by picking an arbitrary state, 
      say $\bar{a} = \ket{0}{}^{\otimes n}$, and construct 
      $L_{\ket{0}{}^{\otimes n} / \ket{1}{}^{\otimes n}}$ as per the procedure
      in Lemma~\ref{lem:Q1SATa}. For each $a \in L_{\ket{0}{}^{\otimes n} / \ket{1}{}^{\otimes n}}$:
      \vspace*{-0.5em}
\begin{itemize}
\item  $a$ defines the region $R_{i, a_i}$ in 
          this branch of the iteration. %\vspace{-0.1em}
\item For $i = 1,\ldots,n$ pick a basis
          $\{\ket{b_i}{}, \ket{b_i^\bot}{}\}$ such that their equator
          bisects $R_{i, a_i}$. %\vspace{-0.1em}
          %of these states 
\item Set $\bar{b} = b_1 \ldots b_n$, construct 
          $L_{\bar{b} / \bar{b}'}$ and for each $b \in
          L_{\bar{b} / \bar{b}'}$: %\vspace{-0.1em}
\begin{itemize}
\item The tuples $(a, b)$ define the region 
              $R_{i, a_ib_i}$ in this branch. %\vspace{-0.1em}
\item Repeat the process to find $\bar{c}$ 
              to bisect each $R_{i, a_ib_i}$; %\vspace{-0.1em}
\item Find a new region $R_{i, a_ib_ic_i}$ 
              for each $c \in L_{\bar{c} /
              \bar{c}'}$. %\vspace{-0.1em}
\item Continue the recursion up to 
              $\log_2 \left(\frac{1}{\epsilon}\right)$ depth and let
              the last list be $L_{z / z^\bot}$. %\vspace{-0.1em}
\item Propose $\ket{\phi^\bot}{} = \bigotimes_{i=0}^n 
              \ket{\varphi_i^\bot}{}$ where $\forall i,
              \ket{\varphi_i}{} \in R_{i,a_ib_i\ldots z_i}$ to the
              oracle. Output $\ket{\phi^\bot}{}$ if the oracle
              returns $\yes$, otherwise continue. 
\end{itemize}
\end{itemize}
\item Output $\unsat$ if none of the trials 
      satisfy the instance. %\vspace*{-0.5em} 
\end{itemize}
%\vspace{-0.2em} 
This algorithm essentially creates a recursion tree with each new
string created where the width of the recursion at each point is
$2mn$ and the depth is $\log_2 \left(\frac{1}{\epsilon}\right)$.
This leads to $(2mn)^{\log_2 \frac{1}{\epsilon}}$ trials to be
proposed at the end and the number of lists created is also
$(2mn)^{\log_2 \frac{1}{\epsilon}}$, each at a cost of $O(mn^2)$.
Hence, the total running time of the algorithm is $O((2mn)^{\log
\frac{1}{\epsilon}} \cdot mn^2)$.

\begin{center}
\begin{figure}[h!]
 	\begin{minipage}[b]{0.49\textwidth}
  	\centering
  	\def\svgwidth{110pt}
  	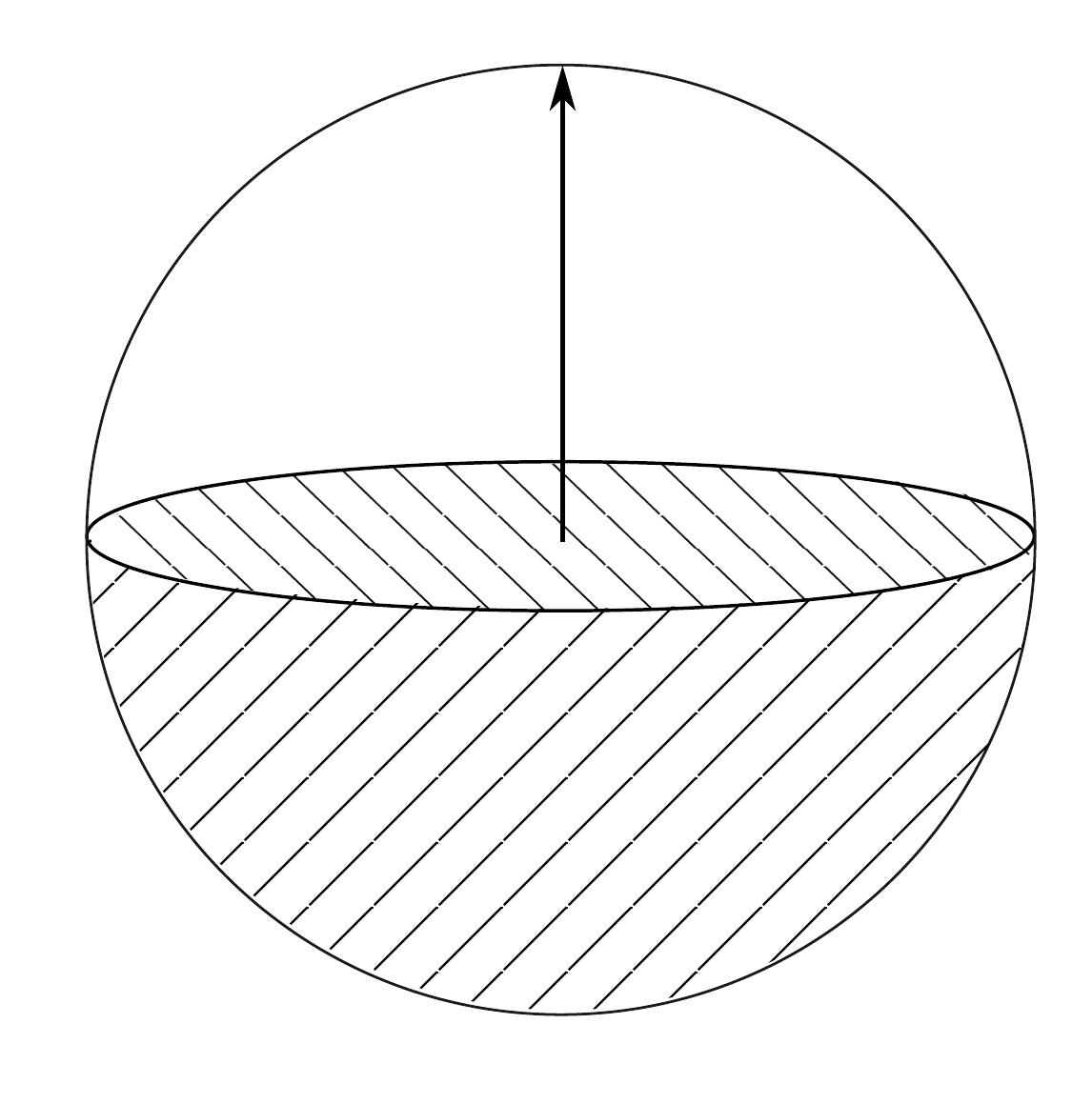
  	\caption*{$(a)$}
%  	\label{fig:R1}
	\end{minipage}
	\begin{minipage}[b]{0.49\textwidth}
  	\centering
  	\def\svgwidth{160pt}
  	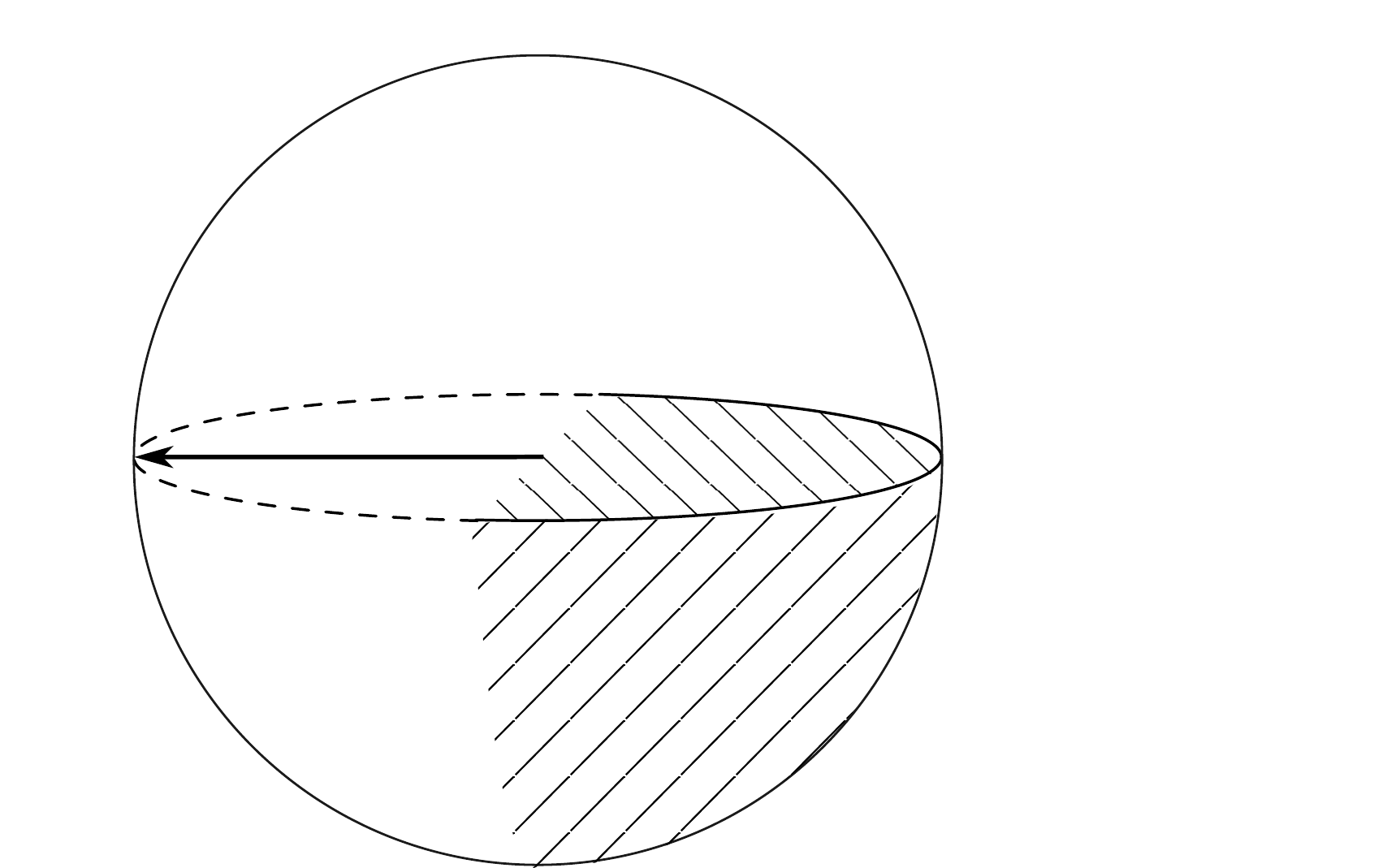
  	\caption*{$(b)$}
%  	\label{fig:R2}
	\end{minipage}

	\vspace*{1em}	
	
	\begin{minipage}[b]{0.49\textwidth}
  	\centering
  	\def\svgwidth{160pt}
  	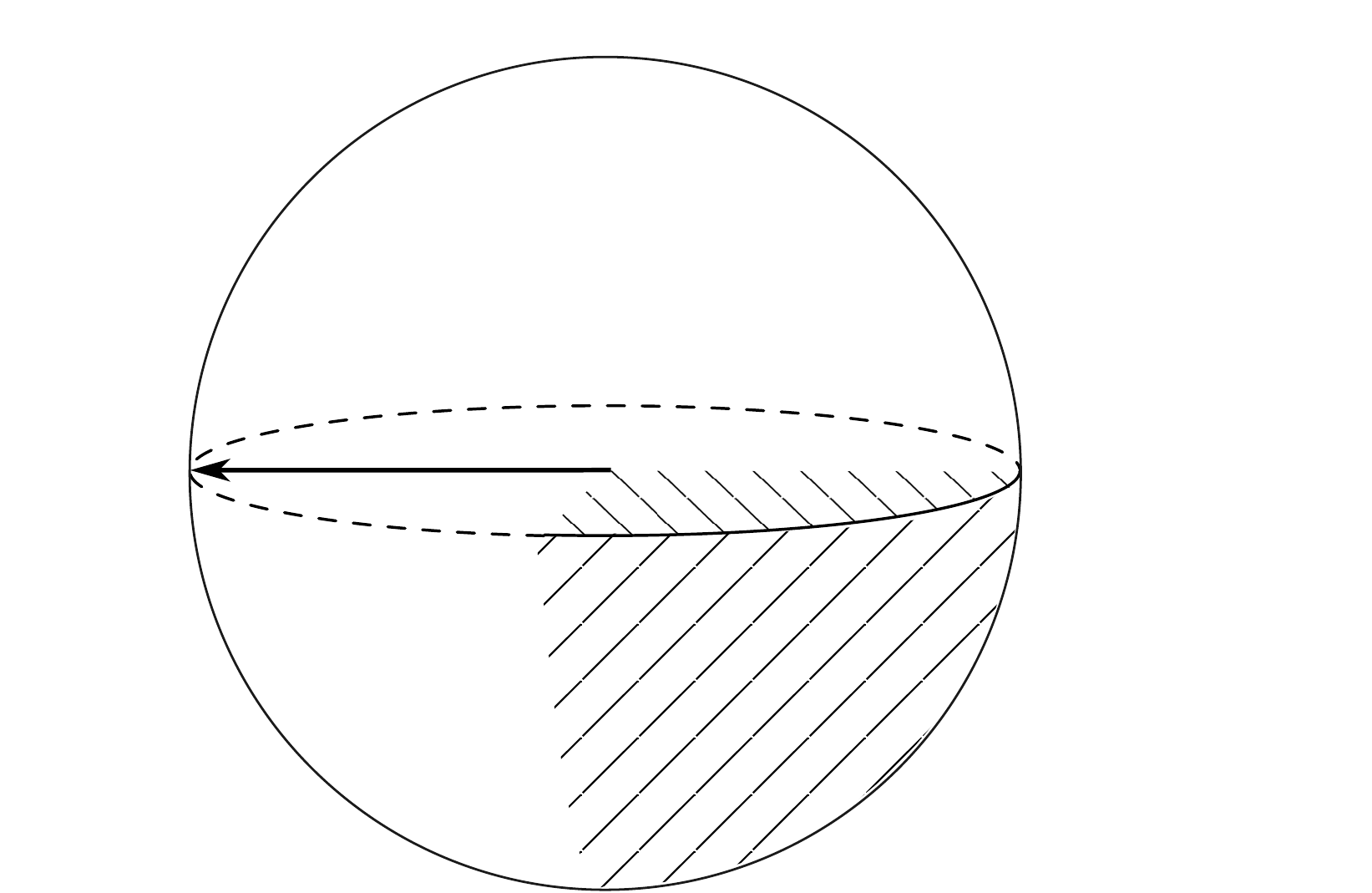
  	\caption*{$(c)$}
%  	\label{fig:R3}
	\end{minipage}	
	\begin{minipage}[b]{0.49\textwidth}
  	\centering
  	\def\svgwidth{120pt}
  	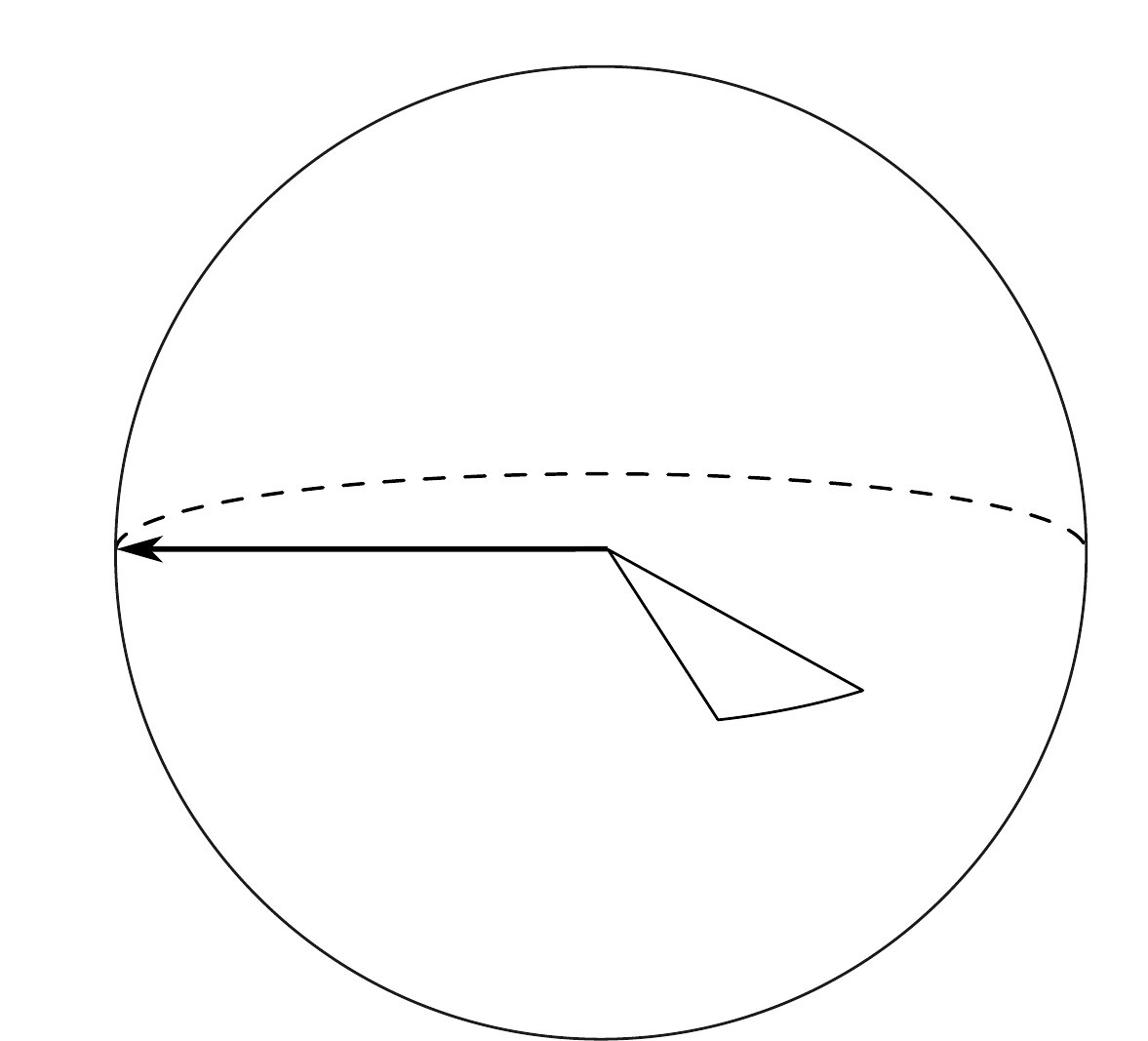
  	\caption*{$(d)$}
%  	\label{fig:Q1SATb}
%  	\captionof{figure}{Finding an $\epsilon$ sized region}
	\end{minipage}
\caption{\label{fig:Q1SAT} Shown here are four stages of the algorithm for a qubit $i$, with forbidden state $\ket{\psi_i}{}$, starting with picking string $a, b$ and $c$ followed by the last string $z$. $(a)$ After picking $\ket{a_i}{}$ for qubit $i$, the hemisphere orthogonal to it is $R_{i, a_i}$; $(b)$ On choosing $\ket{b_i}{}$, the interesting region is the quadrant $R_{i, a_ib_i}$; $(c)$ After $\ket{c_i}{}$ is determined, $\ket{\psi_i}{}$ should be in the hatched $1/8^{th}$ region of the sphere; $(d)$ Continuing the process of picking strings for $\log_2 \frac{1}{\epsilon}$ steps and picking the final state $\ket{z_i}{}$ shows that $\ket{\psi_i}{}$ should be present in $R_{i, a_ib_i\ldots z_i}$.}
\end{figure}
\end{center}

Figure~\ref{fig:Q1SAT} shows exactly how the algorithm given in 
Theorem~\ref{thm:Q1SATb} proceeds. We consider a qubit $i$ and 
the strings $a, b, c, \ldots, z$ that are picked in one branch 
of the recursion tree of the algorithm. To argue the correctness 
of this algorithm, we analyze region $R_{i, a_ib_i\ldots z_i}$ 
obtained at the leaf of the recursion tree. Let the forbidden state 
for qubit $i$ be $\ket{\psi_i}{}$. At the beginning, let 
$\forall \; i, \; |R_i| = 1$ (the complete Bloch sphere) and the only 
guarantee for each list is that there is at least one good string in it. 
Tracing the path in the recursion tree to the leaf, let us assume 
that each step of the recursion picks a good string i.e. $a,b,\ldots,z$ 
are all good strings. For $a$ and $\forall i,$ the forbidden state 
$\ket{\psi_i}{}$ is in the opposite hemisphere to $\ket{a_i}{}$ which 
reduces the size of the region to $|R_{i, a_i}| = 1/2$ as shown in 
Figure~\ref{fig:Q1SAT}\textcolor{Salmon}{$(a)$}. Taking $(a,b)$ at the next iteration, the
region for each qubit is the over lap of two hemispheres $R_{i, a_i}
\cap R_{i, b_i}$ and by construction, since $b_i$ bisects $R_{i,
a_i}$, the overlaps of the hemispheres also bisect $R_{i, a_i}$
setting $|R_{i, a_ib_i}| = 1/4$ (Figure~\ref{fig:Q1SAT}\textcolor{Salmon}{$(b)$}). 
As this pattern continues, each
step of the iteration halves the region for qubit $i$ and we are
left with regions of size at most $\epsilon$ at the end of the
branch as shown in Figures~\ref{fig:Q1SAT}\textcolor{Salmon}{$(c)$} 
and~\ref{fig:Q1SAT}\textcolor{Salmon}{$(d)$}. 
If the instance is satisfiable, the state proposed will satisfy each projector up to
$\epsilon$ resulting in the oracle to return $\yes$. Of course, when
one of the strings chosen is bad, the proposal $\ket{\phi^\bot_j}{}$
for some qubit $j$ will end up having a large inner product with the
forbidden state $\ket{\psi_j}{}$ and will result in the oracle
returning the $\id$ of the projector involving $j$. This concludes
the proof.
\end{proof}

%At the end of this branch of recursion, we pick a $\ket{\phi}{} \in R_{i, a_ib_i\ldots z_i}$ and propose $\ket{\phi^\perp}{}$ for qubit $i$ in this trial. If all the strings chosen were good, since $|R_{i, a_ib_i\ldots z_i}| = O(\epsilon)$, we will get a state such that $\qip{\psi_i}{\phi^\perp} \leq O(\epsilon)$ for qubit $i$.
%\vspace{-1.5em} 
\section{Hidden Quantum 2SAT} \label{sec:2qsat} 
This section deals with a $\QSAT{2}$ instance that is hidden and can
only be accessed by a worst-violation oracle. 
%The setup for $\HQSAT{2}$ is given in Section~\ref{sec:hqsat}. 
We show how learn the underlying local
Hamiltonian to precision $\epsilon$ by finding $2$-local projectors
$\Pi'_{e}$ such that $\norm{ \Pi_{e}-\Pi'_{e}}\le \epsilon$
 for every projector  $\Pi_e$. This
yields an approximate local Hamiltonian $H'=\sum_e \Pi'_e$ whose
ground energy is at most $m\epsilon$ away from the ground
energy of the original Hamiltonian $H=\sum_e \Pi_e$. If $\epsilon$ is set such that $m\epsilon$ is much smaller than the promise gap of the initial Hamiltonian $H$ (which merely requires $\epsilon<1/\operatorname*{poly}$), then the Hamiltonian $H'$ will have a promise gap as well. This is stated in Theorem~\ref{thm:Q2SATa}. 

\begin{theorem}
\label{thm:Q2SATa} 
  Given a $\HQSAT{2}$ problem $H = \sum_{(u, v)} \Pi_{uv}$ on $n$
  qubits, and precision $\epsilon$. If the interaction graph for
  $H$ is not Star-like, then there is an $O(n^4 + n^2 \log
  \left(\frac{1}{\epsilon} \right))$ algorithm that can find an
  approximation $H' = \sum_{(u, v)} \Pi'_{uv}$ where 
  $\forall \; (u, v),  \norm{ \Pi'_{uv}-\Pi_{uv}} 
      \le \epsilon \vspace{-0.5em}$
%  \vspace{-0.5em}
%  \begin{align}
%    
%  \end{align}
%  When $\epsilon \geq \frac{1}{n^\alpha}$ for some constant $\alpha
%  > 0$, we obtain a polynomial time algorithm. 
\end{theorem}

The algorithm proceeds by: 
%\begin{inparaenum}[\upshape(1\upshape)]
\begin{enumerate}[(1)]
\item Identifying two pairs of qubits $(i, j) \ne (k, \ell)$ on 
  which two projectors are defined, $\Pi_{ij}$ and $\Pi_{k\ell}$,
  and finding a constant approximation for these projectors;
\item Improving the constant approximation of the two projectors 
   recursively so that the
  approximation improves by a factor of $2$ in each iteration and
\item Using the $\epsilon$-approximation of a projector to 
  identify the rest of the independent projectors and approximating
  them to $\epsilon$-precision.
\end{enumerate}
%\end{inparaenum}

Using the $H'$ output by the above algorithm in a procedure
which could find a good approximation to the
ground energy of $H'$ would completely solve $\HQSAT{2}$. 
At this time, though, existing $\QSAT{2}$ algorithms~\cite{ASSZ15, Bra06, dBG15} are not 
robust to such errors and seem to require $\frac{1}{\exp(n)}$ precision. 
Our algorithm for $\HQSAT{2}$ does allow one to learn the projectors to exponential precision, since the dependence on $\epsilon$ in Theorem~\ref{thm:Q2SATa} is merely logarithmic. However, in this parameter regime our algorithm is somewhat unrealistic, as this would require the oracle to be able to distinguish between values that are exponentially close together\footnote{This seems to give the oracle too much power - because having the ability to distinguish exponentially close quantum states would allow one to solve $\PP$-hard problems \cite{AbramsLloyd}. 
In contrast all of the problems considered are in $\NP$ due to the presence of classical, $\poly(n)$ size witnesses.}. 
If our oracle were constrained to be implementable in polynomial time by an experimenter, acting on polynomially many copies of the proposed state $\rho$, then one could only learn the instance up to error $\epsilon = \frac{1}{\poly(n)}$.
A natural open question is to determine whether one can still solve $\QSAT{2}$ when one only knows the individual clauses to inverse polynomial precision; we believe this is a fundamental question about the nature of $\QSAT{2}$, which is left for future work.

To prove Theorem~\ref{thm:Q2SATa}, we require the notion of converting energy violations to distances between states and vica-versa as presented below.

\vspace{-0.5em}
\subparagraph*{Energies and Distances.} Along with the energy of a state with respect to a local term, another useful measure of violation is the distance between quantum states which is helpful in understanding how precise the solution is. We will often switch between distances of $2$-qubit states and their violation energy. It is therefore important to relate these two measures.

Assume we have a projector $\Pi_\psi=\bk{\psi}{}$ and a state $\rho_\alpha\EqDef \bk{\alpha}{}$. Recall that the violation energy is given by $\Tr(\Pi_\psi\rho_\alpha) = |\qip{\alpha}{\psi}|^2$. The \emph{Frobenius distance} is defined as 
\begin{align}
\norm{\alpha-\psi} \EqDef \sqrt{\Tr[ (\bk{\alpha}{} - \bk{\psi}{}) (\bk{\alpha}{} - \bk{\psi}{})^\dagger]}.
\end{align}
It now follows that $\norm{\alpha-\psi}^2 = 2 - 2|\qip{\alpha}{\psi}|^2 = 2-2\Tr(\Pi_\psi\rho_\alpha)$. Therefore, the violation energy is related to the Frobenius distance as 
\begin{align}
\label{eq:energy-distance}
 \Tr(\Pi_\psi\rho_\alpha) = 1-\frac{1}{2}\norm{\alpha-\psi}^2 \,.
\end{align}

For the sake of clarity, we currently assume that all the projectors in $H'$ are of rank $1$ and discuss the entire algorithm for this case. To generalize to cases where the rank of projectors is $>1$ requires only a slight modification that does not affect the running time significantly and will be sketched after describing the algorithm in full. Now, we prove a series of lemmas that will enable us to prove Theorem~\ref{thm:Q2SATa}. 

Before delving into the details of the implementation for Step $1$, we construct a procedure \textbf{Test-I}, which checks if there is a projector between particles $i,j$ at a constant distance from $\bk{\psi}{}$. The test returns $\yes$ if and only if there is a projector at $i,j$ which is close to $\bk{\psi}{}$ along with the $\id$ of the projector and $\no$ otherwise. The test is defined for two fixed constants $0<\nu<1$ and $0<\epsilon<1$, and a $\nu$-net $\cc{T}_\nu$ in the space of rank-$1$ projectors on two qubits.

\noindent \textbf{Test-I$\mathbf{(i,j,\bk{\psi}{})}$}
\begin{itemize}
  \item For all possible pairs $(k, \ell)$ which are different from $(i, j)$:
  \begin{itemize}
  	\item[-] For all $\bk{\alpha}{}\in \cc{T}_\nu$, construct the trial state $\rho^{k\ell}_\alpha$, propose it to the oracle and receive a projector $\id$ as violation where
    	\begin{align}
      		\rho^{k\ell}_\alpha \EqDef \bk{\psi}{ij}\otimes \big[ (1-\epsilon) \bk{\alpha}{k\ell} +
        \epsilon\bk{\alpha^\bot}{k\ell}\big] \otimes \left(\frac{\mathbb{I}}{2}\right)^{\otimes (n-4)}
	    \end{align}
  \end{itemize}
  \item If for all trials the oracle answers the \emph{same} $\id$, output
    $\yes$ together with the $\id$ of the projector, otherwise output $\no$.
\end{itemize}

The correctness of \textbf{Test-I} is proved in the following claim.
\begin{clm}
\label{clm:Test-I}
  Assume that $\epsilon+\frac{1}{2}\nu^2 < 1/4$, and further assume
  that there is at least one projector in the system that is defined
  on qubits $k\ne \ell$, which are \emph{different} from $i,j$. Then,
  if Test-I outputs $\yes$, there exists a projector
  $\bk{\psi'}{}$ on the $(i,j)$ qubits, whose $\id$ is the one that 
  was output and
  \begin{align}
    \norm{\bk{\psi'}{}-\bk{\psi}{}} \le \sqrt{2\epsilon+\nu^2} \,.
  \end{align}
  Conversely, if there exists a projector $\bk{\psi'}{}$ on $(i,j)$
  such that   
  \begin{align}
    \norm{\bk{\psi'}{}-\bk{\psi}{}} < \sqrt{2\epsilon} \,,
  \end{align}
  then the test will report it.
\end{clm}

\begin{proof}
  Assume we have a projector $\bk{\phi'}{}$ on qubits $k,\ell$, which are different than $i,j$. Let $\bk{\alpha}{}$ be the closest member of $\cc{T}_\nu$ to $\bk{\phi'}{}$, and let us calculate the violation energy for this particular assignment. Note that  $\Tr(\bk{\alpha^\bot}{}\cdot\bk{\phi'}{})= 1-\Tr(\bk{\alpha}{}\cdot\bk{\phi'}{})$ and so the violation due to this assignment is 
  \begin{align*}
    &(1-\epsilon)\Tr(\bk{\alpha}{}\cdot\bk{\phi'}{})
      + \epsilon \big[1-\Tr(\bk{\alpha}{}\cdot\bk{\phi'}{})\big] \\
       &= (1-2\epsilon)\Tr(\bk{\alpha}{}\cdot\bk{\phi'}{}) +
       \epsilon
  \end{align*}
  By assumption, $\norm{\bk{\alpha}{}-\bk{\phi'}{}}\le \nu$, and therefore by \Eq{eq:energy-distance},
  \begin{align*}
    1-\frac{1}{2}\nu^2 \le \Tr(\bk{\alpha}{}\cdot\bk{\phi'}{}) 
      \le 1\,,
  \end{align*}
  which implies that
  \begin{align*}
    1-\epsilon-\frac{1}{2}\nu^2\le \text{max $(k,\ell)$ violation energy} 
      \le 1-\epsilon \,.
  \end{align*}
  
  Similarly, by looking at the state from the $\nu$-net that is closest to $\bk{\phi'^\bot}{}$, we deduce that
  \begin{align*}
    \epsilon\le \text{min $(k,\ell)$ violation energy} 
      \le \epsilon+\frac{1}{2}\nu^2 \,.
  \end{align*}
  
  If we got the same answer for all tests then it cannot be due to one of the completely mixed states since there the 
  violation is always $1/4$, and we know that for the maximal $\alpha$ violation is at least 
  $1-\epsilon-\frac{1}{2}\nu^2> 1/4$. It cannot also be due to the $(k,\ell)$ projector since we know that the minimal
  violation energy there is at most $\epsilon+\frac{1}{2}\nu^2<1/4$, i.e., less violated than the mixed projectors.
  Therefore, it must be the projector at $(i,j)$.
  
  Moreover, the violation of $(i,j)$ must be at least as big as the maximal $(k,\ell)$ violation:
	\begin{align*}
	\Tr(\bk{\psi}{}\cdot\bk{\psi'}{}) \geq 1-\epsilon-\frac{1}{2}\nu^2 \,.
  \end{align*}
	
  Therefore,
  \begin{align*}
	\norm{\bk{\psi}{}-\bk{\psi'}{}} = \sqrt{2-2\Tr(\bk{\alpha}{}\cdot\bk{\psi'}{})} \le \sqrt{2\epsilon + \nu^2} \,.
  \end{align*}
  
  For the other direction, note that if
  $\norm{\bk{\psi}{}-\bk{\psi'}{}}<\sqrt{2\epsilon}$ then by \Eq{eq:energy-distance}, its violation must satisfy
  $\Tr(\bk{\psi}{}\cdot\bk{\psi'}{})> 1-\epsilon$, which is bigger than both the violations of the completely 
  mixed state and the maximal violation of $(k,\ell)$.
\end{proof}

We call a $2$-qubit state $\ket{\psi}{ij}$ \textit{$\delta$-good} for $\delta = \sqrt{2\epsilon + \nu^2}$ if it is returned during a call to \textbf{Test-I}$(i, j, \bk{\psi}{})$. The idea is that if we find an \textit{$\delta$-good} state for $(i, j)$, it would be a constant approximation for the projector on $(i, j)$. Suppose no $\epsilon$-good state is found for $(i, j)$ then we can conclude that there is \emph{no projector} on $(i, j)$ or the interaction graph for $H$ has a \emph{Star-like} configuration.

\noindent \textbf{Step $1$}
\begin{itemize}
  \item Set $\nu^2 \EqDef\frac{1}{16}$, $\epsilon\EqDef \frac{1}{32}$ and $\eta \EqDef \frac{1}{4}$  
  \item Repeat for all pairs $i \ne j$ until a \textit{$\delta$-good} approximation for some $\Pi_{ij}$ is found.
  \begin{itemize}
  	\item For all $\ket{\psi}{}\in \cc{T}_\eta$, perform \textbf{Test-I}$(i,j, \bk{\psi}{})$. If the test is positive, set $\bk{\psi}{}$ as the \textit{$\delta$-good} approximation for $\Pi_{ij}$.
  \end{itemize}
  \item Repeat the above process for all qubit pairs $(k, \ell)$ that are disjoint from $(i, j)$ till a 
  \textit{$\delta$-good} approximation for $\Pi_{k\ell}$ is found.
\end{itemize}
  
\begin{lemma}
\label{lem:Step1}
  If there exist two projectors in $H'$ acting on disjoint pairs of qubits then Step $1$ will always succeed in finding   independent pairs $(i,j), (k, \ell)$ and projectors $\bk{\psi^{(0)}}{ij}$, $\bk{\phi^{(0)}}{k\ell}$ such that $\norm{\bk{\psi^{(0)}}{ij} - \bk{\psi'}{ij}} \leq \frac{1}{\sqrt{8}}$ and $\norm{\bk{\phi^{(0)}}{k\ell} - \bk{\phi'}{k\ell}} \leq \frac{1}{\sqrt{8}}$. Moreover, Step $1$ can be executed in $O(n^4)$ time. 
\end{lemma}

\begin{proof}
  Set the parameters according to Step $1$. Then from Claim~\ref{clm:Test-I}, \textbf{Test-I} succeeds in finding states that are $\sqrt{2\epsilon + \nu^2}$-good $= \frac{1}{\sqrt{8}}$-good approximations. For each \textbf{Test-I}$(i, j, \ket{\psi})$, iterating over all possible disjoint pairs $(c, d)$ gives us $\binom{n-2}{2}$ pairs to check. The space of $2$-qubit states being in $\mathbb{C}^2 \otimes \mathbb{C}^2$ is a space of dimension $4$ and for any $\gamma > 0$, a $\gamma$-net over $\mathbb{C}^2 \otimes \mathbb{C}^2$ will contain $O\left(\frac{1}{\gamma^4}\right)$ points. So, for $\nu, \eta \in O(1)$, we check only $O(1)$ states in $\cc{T}_\nu$ for each qubit pair $(c, d)$ and run \textbf{Test-I} for $O(1)$ states in $\cc{T}_\eta$ giving a total of $O(n^2)$ trials proposed. To find the second projector on $(k, \ell)$ we repeat the above process for at most $O(n^2)$ pairs. Hence, Step $1$ can be executed in $O(n^4)$ time and the result follows directly.
\end{proof}

At the end of Step $1$ we have two projectors $\bk{\psi^{(0)}}{ij}$ and $\bk{\phi^{(0)}}{k\ell}$ which are 
a constant approximation of their hidden counterparts $\Pi'_{ij}$ and $\Pi'_{k\ell}$ i.e. 
at a distance $\delta \leq \sqrt{2\epsilon + \nu^2}$. For this step, we 
construct the procedure \textbf{Test-II} to improve their value to 
$\bk{\psi^{(1)}}{ij}$ (resp. $\bk{\phi^{(1)}}{k\ell}$) such that it is at a distance 
$\leq \frac{\delta}{2}$ from $\Pi'_{ij}$ (resp $\Pi'_{k\ell}$). Then, Step $2$ basically repeats this test 
$c$ times to improve the value to $\Pi^{(c)}_{ij}$ to a distance $\leq \frac{\delta}{2^c}$ 
from $\Pi'_{ij}$ and when $c = O(\log{n})$, this will get us polynomially close to $\Pi'_{ij}$.

From Step $1$, we know that $\Pi'_{ij}$ lies somewhere in a radius of $\delta$ around $\bk{\psi^{(0)}}{ij}$ and similarly for qubits $(k, \ell)$. Let $\cc{B}_{ij}$ be the ball of radius $\delta$ around $\bk{\psi^{(0)}}{ij}$ and correspondingly consider $\cc{B}_{k\ell}$. The states will be enumerated over $\cc{T}_{\nu'}^{\cc{B}}$ which is the $\nu'$-net restricted to some ball $\cc{B}$ in the space of rank-$1$ projectors on $2$ qubits.

Set $\nu' = \frac{\nu}{2}$, $\eta' = \frac{\eta}{2}$ and $\epsilon' = \frac{\epsilon}{4}$. The test is defined for values of $\nu', \eta', \epsilon' > 0$ as\\

\noindent \textbf{Test-II}$(i, j, k, \ell)$
\begin{itemize}
	\item For $(i, j)$, repeat over all $\bk{\psi}{} \in \cc{T}_{\eta'}^{\cc{B}_{ij}}$ and perform 
	\textbf{Test-I}$(i, j, \bk{\psi}{})$ over $\cc{T}_{\nu'}^{\cc{B}_{k\ell}}$ with parameter $\epsilon'$. If the test outputs $\yes$, set $\ket{\psi^{(1)}}{ij} = \ket{\psi}{}$.
%	\item If the test outputs $\yes$, set $\ket{\psi^{(1)}}{ij} = \ket{\psi}{}$
	\item For $(k, \ell)$, repeat over all $\bk{\phi}{} \in \cc{T}_{\eta'}^{\cc{B}_{k\ell}}$ and perform 
	\textbf{Test-I}$(k, \ell, \bk{\phi}{})$ over $\cc{T}_{\nu'}^{\cc{B}_{ij}}$ with parameter $\epsilon'$. If the test outputs $\yes$, set $\ket{\phi^{(1)}}{ij} = \ket{\phi}{}$
%	\item If the test outputs $\yes$, set $\ket{\phi^{(1)}}{ij} = \ket{\phi}{}$
\end{itemize}

The correctness of \textbf{Test-II} is determined by Claim~\ref{clm:Test-II}.

\begin{clm}
\label{clm:Test-II}
If there exists projectors $\bk{\psi^{(l)}}{ij}$, $\bk{\phi^{(l)}}{k\ell}$ on qubit pairs $(i, j)$, $(k, \ell)$ that have been approximated to a distance $\delta$, then \textbf{Test-II}$(i, j, k, \ell)$ will give us projectors 
$\bk{\psi^{(l+1)}}{ij}$ and  $\bk{\phi^{(l+1)}}{k\ell}$ such that 
\begin{align}
\norm{\bk{\psi^{(l+1)}}{ij} - \bk{\psi'}{ij}} \leq \frac{\delta}{2} \quad \text{ and } \quad
\norm{\bk{\phi^{(l+1)}}{k\ell} - \bk{\phi'}{k\ell}} \leq \frac{\delta}{2}.
\end{align}
\textbf{Test-II}$(i, j, k, l)$ requires $O\left( \left(\frac{\delta}{\eta'}\right)^4 \times \left(\frac{\delta}{\nu'}\right)^4 \right)$ trials where $\delta$ is the radius of $\cc{B}_{ij}$ and $\cc{B}_{k\ell}$.
\end{clm}

\begin{proof}
Clearly, as \textbf{Test-I} works on the complete space of $2$ qubit rank $1$ projectors, it will also work on the restricted ball of size $\delta$. Since the existence of nontrivial projectors $\Pi'_{ij}$ and $\Pi'_{k\ell}$ has already been determined from Step $1$, we are sure to find another state over the new $\eta'$-net that approximates the projectors according to the new values. Then, setting the parameters as mentioned in Step $2$, the bound follows directly from Claim~\ref{clm:Test-I}. For the number of trials, since $(k, \ell)$ is fixed for $(i, j)$ and vica-versa, the trials only  iterate over the number of states in the $\delta$-ball of a $\gamma$-net over the space of $2$-qubit states which contains $O\left(\frac{\delta^4}{\gamma^4}\right)$ states. Substituting for the values of $\gamma$ gives the required number of trials.
\end{proof}\\

\noindent \textbf{Step $2$}\\
To approximate the projectors on qubit pairs $(i, j), (k, \ell)$ to polynomial accuracy, collect parameters $\eta, \nu$ and $\epsilon$ from Step $1$. Set the counter $c = 0$.
\begin{itemize}
	\item Set $\nu' = \frac{\nu}{2}$, $\eta' = \frac{\eta}{2}, \epsilon' = \frac{\epsilon}{4}$ and 
	$\delta = \sqrt{2\epsilon + \nu^2}$.
	\item Let $\cc{B}_{ij}$ be the ball of radius $\delta$ around $\bk{\psi^{(c)}}{ij}$ and correspondingly 
	$\cc{B}_{k\ell}$.
	\item Run \textbf{Test-II}$(i, j, k, \ell)$ which output states $\bk{\psi^{(c+1)}}{ij}$ and $\bk{\phi^{(c+1)}}{k\ell}$
	\item Update the parameters such that $\nu = \nu'$, $\eta = \eta'$ and $\epsilon = \epsilon'$.
	\item Increment the counter and repeat the process till $c = O(\log{n})$.
\end{itemize}

A crucial part of the analysis for Step $2$ is to show that it can be executed in polynomial time. This is ensured by showing that the number of trials for each iteration of \textbf{Test-II} actually remains a constant independent of $n$. 

\begin{lemma}
\label{lem:Step2}
Given projectors $\bk{\psi^{(1)}}{ij}, \bk{\phi^{(1)}}{k\ell}$ on qubit pairs $(i, j), (k, \ell)$ that have been approximated to a distance $\delta$, Step $2$ will successfully find projectors $\bk{\bar{\psi}}{ij}$ and $\bk{\bar{\phi}}{k\ell}$ such that $\norm{\bk{\bar{\psi}}{ij} - \bk{\psi'}{ij}} \leq \frac{1}{\poly(n)}$ and $\norm{\bk{\bar{\phi}}{k\ell} - \bk{\phi'}{k\ell}} \leq \frac{1}{\poly(n)}$. Additionally, this step can be executed in $O(\log n)$ time.
\end{lemma}

\begin{proof}
We start the first iteration in Step $2$ with the parameters $\delta, \frac{\nu}{2}, \frac{\eta}{2}$ and proceed in each iteration by halving these parameters. In effect, at the $t^{th}$ iteration, the parameters used are $\frac{\delta}{2^{t-1}}, \frac{\eta}{2^t}$ and $\frac{\nu}{2^t}$. The costliest operation in executing Step $2$ involves \textbf{Test-II} being performed at each iteration. From Claim~\eqref{clm:Test-II} the $t^{th}$ iteration of \textbf{Test-II} can be executed in time expressed via the parameters $\delta_t, \nu_t, \eta_t$ as
\[
O\left( \frac{\delta_t^8}{\eta_t^4 \nu_t^4} \right)
= O\left( \frac{\delta^8}{2^{8(t-1)}} \frac{2^{4t}}{\eta^4} \frac{2^{4t}}{\nu^4} \right)
= O\left( \frac{\delta^8} {\eta_t^4 \nu_t^4} 2^{8} \right) \in O(1)
\]
where the last inclusion holds as $\delta, \eta$ and $\nu$ start as constants. With $O(\log n)$ iterations, this leads to Step $2$ being executed in $O(\log n)$ time. Considering the accuracy of the states output, it is clear that in iteration $t$ the projectors output are $\frac{\delta}{2^t}$ close to the projectors in the hidden instance. For $t = O(\log n)$ this translates to a distance of $\frac{\delta}{2^{O(\log n)}} \leq \frac{\delta}{O(n^c)}$ for some constant $c > 0$ and this in turn is written as $\frac{1}{\poly(n)}$ for $\delta \in O(1)$ and the result follows.
\end{proof}\\

\noindent \textbf{Step $3$}\\
To approximate the remaining projectors, do the following:
\begin{itemize}
	\item Pick a pair of qubits $(u, v)$ that is independent from at least one of the projectors approximated so far. 
	\item Perform Step $1$ to approximate $\Pi_{uv}^{(0)}$ to constant accuracy (if it exists). Otherwise, pick another pair of qubits.
	\item Let $(x, y)$ be independent of $(u, v)$ such that $\Pi_{xy}$ has been approximated to $\frac{1}{\poly(n)}$ accuracy. Use $\Pi_{xy}^{(0)}, \ldots, \Pi_{xy}^{(O(\log n))}$ to approximate $\Pi_{uv}$ to $\frac{1}{\poly(n)}$ accuracy as per Step $2$.
	\item Repeat for all possible independent qubit pairs\footnote{To learn a projector $\Pi_{ik}$ when $\Pi_{ij}$ has already been found, set qubit $j$ to the mixed state and choose a different $\Pi_{mn}$ to use in Steps $1$ and $2$ for improving the accuracy of $\Pi_{ik}$}.
\end{itemize}

\vspace*{-1.5em}
\subparagraph*{Dealing with higher rank projectors.} As mentioned earlier, the tests have been clearly designed assuming the presence of only rank $1$ projectors. To generalize \textbf{Test-I} for projectors of rank $> 2$, we don't stop after finding just one state $\ket{\psi}{}$ that succeeds the test. By continuing to iterate over all $2$-qubit states, we can find a constant number of states that span the forbidden subspace and then use any one of the linear algebra techniques to find a basis for that space whose dimensions would give the rank of the projector. This would also approximate the basis up to constant accuracy at the end of \textbf{Test-I} as the states that will be returned from the test can be shown as having a low distance (or high violation energy) with at least one of the non-zero components of the high rank projector. Then, repeating \textbf{Test-II} for each basis element will successfully approximate each of them to $\frac{1}{\poly(n)}$ accuracy. Note that each of these changes do not significantly affect the runtime of the algorithm or the number of trials proposed. 

Now we proceed to the proof of Theorem~\ref{thm:Q2SATa}.\\

\noindent \textbf{Proof of Theorem~\ref{thm:Q2SATa}}
As outlined previously, putting together the $3$ steps gives us the required algorithm. Consider \textbf{Test-I}$(i, j)$ contains the projector $\ket{\psi}{}$ and all the states of the form $\rho_{k\ell}^{\alpha}$ used for the test. From Claim~\ref{clm:Test-I}, we know that $\norm{\ket{\alpha}{kl} - \Pi'_{kl}} \leq \frac{1}{\sqrt{8}}$ for some $(k, l)$ and some $\alpha$. Then, any state that is output by \textbf{Test-I} should be closer to $\Pi'_{ij}$ to have a larger overlap with it. In case of $(i, j)$ being disjoint from $(k, \ell)$, there is no problem to ensure this. 

However, when no projectors independent of $(i,j)$ exist, consider another projector $(i, k)$. Now, the states used in \textbf{Test-I} would be of the form $\rho_k^\alpha$. Let the reduced density matrix on $i$ with respect to $\bk{\psi}{}$ be $\rho_i$. Then the error threshold used for any state output by \textbf{Test-I} in this case would be related to $\max_{\alpha} \Tr(\rho_i \ket{\alpha}{k} \Pi'_{ik}) \leq \Tr(\rho_i \Tr_k(\Pi'_{ik}) )$. It is possible for the latter value to be almost $0$ in the case that $\rho_i$ almost satisfies $\Pi'_{ik}$ (e.g. the product state on $i$ satisfies $\Pi'_{ik}$). This would lead to an inaccurate error threshold and affect the accuracy of the states output by \textbf{Test-I}. This explains the necessity of $H'$ not having a \emph{Star-like} configuration.

For the running time argument, from Lemmas~\ref{lem:Step1} and~\ref{lem:Step2}, the running time for finding a projector up to constant accuracy is $O(n^2)$ and to improve the accuracy to $\beta << 1$ takes $O(\log \frac{1}{\beta})$ time. Step $3$ essentially repeats the combination of (Step $1$ + Step $2$) for $n^2$ pairs of qubits and results in an overall running time of $O(n^2(n^2 + \log \frac{1}{\beta}))$. Setting $\beta < \frac{1}{n^c}$ for some constant $c$, makes the overall runtime $O(n^4 + n^2 \log n) = O(n^4)$. Similarly, the correctness also follows from combining Lemmas~\ref{lem:Step1} and~\ref{lem:Step2}. $\Box$

Theorem~\ref{thm:Q2SATa} clearly excludes the pathological cases of \emph{Star-like} configurations. These are discussed below.

Learning even one of the projectors in the case of the interaction graph being a Star\footnote{A star graph is one where no two edges of the graph are independent.} is impossible. However, the intermediate case when there is exactly one edge in the graph that is not independent of the other edges, seems to have an intermediate albeit slightly unnatural solution. In fact, to explicitly learn some projectors and then solve the instance requires the oracle to distinguish between exponentially small values. In particular, the following lemma holds.

\begin{lemma}
\label{thm:Q2SATb}
Given a $\HQSAT{2}$ problem $H' = \sum_{(u, v)} \Pi'_{uv}$ on $n$ qubits, $\epsilon \geq \frac{1}{n^\beta}$ for a constant $\beta$ and a function $f(n) \in \exp(n)$. If there is exactly one edge $(i, k)$ that does not have any independent projector, then there is an $O(n^4 + n^2 \log f(n))$ time algorithm that can find an approximation $H = \sum_{(u, v)} \Pi_{uv}$ where 
\[ \forall \; (i, j), \neq (i, k)  \norm{ \Pi'_{ij}-\Pi_{ij}} \le \frac{1}{f(n)}\]
and we can find a $2$ qubit state $\ket{\Phi}{ik}$ in $O(\epsilon^4)$ time, such that $\bra{\Phi}{} \Pi'_{ik} \ket{\Phi}{} \leq \epsilon^2$.
\end{lemma}

\begin{proofs}
For every qubit pair except $(i, k)$, use Step $1$ to find a constant approximation to the hidden projector and repeating Step $2$ for $\poly(n)$ iterations, find exponentially close approximations to the hidden projectors. Using these approximations, we can find an $n-2$-qubit (resp. $n-1$-qubit) state that almost satisfies all projectors except $\Pi'_{ik}$ following any $O(n+m)$ algorithm to solve $\QSAT{2}$~\cite{ASSZ15,dBG15}. The $n-1$ qubit state includes either $i$ or $k$ but not both in the satisfiable case. Then, iterating over all $2$-qubit (resp. $1$-qubit) states on an $\epsilon$-net, and proposing the complete $n$ qubit state to the oracle, will let us find a state that has low overlap with $\Pi'_{ik}$.
\end{proofs}

%\begin{Rem}
For the remaining case of the Star graph, at the present time, we do not have an algorithm to learn the projectors to any level of accuracy. This is due to the interference of the center of the star with all the projectors skewing the error thresholds used in this type of algorithm. Of course, the brute force technique to find the ground state by iterating over an $\epsilon$-net of all $n$-qubit states with at most pairwise entanglement leads to an exponential number of trials to be proposed but the power of the oracle doesn't change. Hence, an obvious trade-off between the power of the oracle and the number of trials proposed exists although both techniques currently lead to unnatural algorithmic techniques for the pathological cases.
%\end{Rem}

\section{Acknowledgements}
Research was supported by the Singapore Ministry of Education and the National Research Foundation by the Tier 3 Grant MOE2012-T3-1-009, by the European Commission IST STREP project Quantum Algorithms (QALGO) 600700, the French ANR Blanc Program Contract ANR-12-BS02-005. A.B. was supported in part by the NSF Graduate Research Fellowship under grant no. 1122374 and by the NSF Alan T. Waterman award under grant no. 1249349. S.Z.'s research was supported in part by RGC of Hong Kong (Project no. CUHK419413).

\newpage
\bibliographystyle{plain}
\bibliography{ref}

\end{document}